\documentclass{tlp}
\usepackage{amsfonts}
\usepackage{amssymb}
\usepackage{latexsym}
\usepackage{stmaryrd}
\usepackage{microtype}
\usepackage{comment}
\usepackage[arxiv]{preamble}
\usepackage[disable]{review}

\ifarxiv
\usepackage{hyperref}
\fi

\usepackage{etoolbox}

\ExplSyntaxOn

\newcommand\setcitation[2]{%
	\csdef{mycommoncitation\text_uppercase:n{#1}}{#2}%
	\csappto{bbAllCommonCitations}{\cite{#2}\ }%%%for testing purpuses only	
}
\newcommand\getcitation[1]{%
	\csuse{mycommoncitation\text_uppercase:n{#1}}}

\newcommand\refto[1]{%
	\ifcsname  mycommoncitation\text_uppercase:n{#1}\endcsname%
	\getcitation{#1}%
	\else%
	#1%
	\fi%
}

\ExplSyntaxOff

\newcommand\mycitep[1]{\citep{\refto{#1}}}

\setcitation{IDP}{DBBJD18PredicatelogicmodelinglanguageIDPsystem}
\setcitation{fodot}{DT08logicnonmonotoneinductivedefinitions}
\setcitation{CPsupport}{DBDD13ModelExpansionPresenceFunctionSymbolsUsing}
\setcitation{minisatid}{DBDD13ModelExpansionPresenceFunctionSymbolsUsing}
\setcitation{FunctionDetection}{DB13Detectionexploitationfunctionaldependenciesmodelgeneration}
\setcitation{lazyGrounding}{DDBS15LazyModelExpansionInterleavingGroundingSearch} 
\setcitation{Bootstrapping}{BJDJBD16BootstrappingInferenceIDPKnowledgeBaseSystem}
\setcitation{BootstrappingIDP}{BJDJBD16BootstrappingInferenceIDPKnowledgeBaseSystem}
\setcitation{FOSymmetry}{DBBD16localdomainsymmetrymodelexpansion}
\setcitation{FOLASP}{VDV21FOLASPFOInputLanguageAnswerSet}
\setcitation{inputster}{JJJ13CompilingInputFOinductivedefinitionsinto}
\setcitation{LearningPaper}{BBBDDJLR15Predicatelogicmodelinglanguagemodelingsolving}

\setcitation{foid}{DT08logicnonmonotoneinductivedefinitions}
\setcitation{cplogic}{VDB10FOIDextensionDLrules}
\setcitation{clog}{BVDV14FOCKnowledgeRepresentationLanguageCausality} %TODO replace by journal publication if one is published
\setcitation{foc}{BVDV14FOCKnowledgeRepresentationLanguageCausality}
\setcitation{inferenceClog}{BVDV14InferenceFOCModellingLanguage}
\setcitation{inferenceFOC}{BVDV14InferenceFOCModellingLanguage}
\setcitation{examplesClog}{BVDV14FOCRelatedModellingParadigms} %TODO replace by
\setcitation{satid}{MWDB08SATIDSatisfiabilityPropositionalLogicExtendedInductive}

\setcitation{fodot2asp}{DLTV19informalsemanticsAnswerSetProgrammingTarskian} %TODO replace by journal publication if one is publishedr
\setcitation{Tarskian}{DLTV19informalsemanticsAnswerSetProgrammingTarskian} %Same as the one above 
\setcitation{TarskianSemanticsASP}{DLTV19informalsemanticsAnswerSetProgrammingTarskian} %Same as the one above 
\setcitation{KBS}{DV08BuildingKnowledgeBaseSystemIntegrationLogic}
\setcitation{KBS-invitedtalk}{D16FOKnowledgeBaseSystemproject}
\setcitation{KBPE}{DWD11PrototypeKnowledge-BasedProgrammingEnvironment}

\setcitation{justifications}{DBS15FormalTheoryJustifications} 
\setcitation{justificationTheory}{DBS15FormalTheoryJustifications} 
\setcitation{AFT}{DMT00ApproximationsStableOperatorsWell-FoundedFixpointsApplications}

\setcitation{SAT}{BHvW21HandbookSatisfiability-SecondEdition}
\setcitation{HandbookOfSAT}{BHvW21HandbookSatisfiability-SecondEdition}
\setcitation{SMS}{KS21SATModuloSymmetriesGraphGeneration}

\setcitation{totalizer}{BB03EfficientCNFEncodingBooleanCardinalityConstraints}

\setcitation{DPLLT}{GHNOT04DPLLTFastDecisionProcedures}
\setcitation{SMT}{BSST21SatisfiabilityModuloTheories}

\setcitation{CP}{RvW06HandbookConstraintProgramming}
\setcitation{csp2asp}{DW11Translation-BasedConstraintAnswerSetSolving}
\setcitation{EZCSP}{B11Industrial-SizeSchedulingASPCP}
\setcitation{Inca}{DW12AnswerSetSolvingLazyNogoodGeneration}
\setcitation{Zinc}{MNRSGW08DesignZincModellingLanguage}
\setcitation{MiniZinc}{NSBBDT07MiniZincTowardsStandardCPModellingLanguage}

\setcitation{ASPComp2}{DVBGT09SecondAnswerSetProgrammingCompetition}
\setcitation{ASPComp3}{CIR14thirdopenanswersetprogrammingcompetition}
\setcitation{ASPComp4}{ACCDDIKK13FourthAnswerSetProgrammingCompetitionPreliminary}
\setcitation{ASPComp5}{CGMR16DesignresultsFifthAnswerSetProgramming}
\setcitation{ASPComp6}{GMR17SixthAnswerSetProgrammingCompetition}
\setcitation{ASPComp7}{GMR20SeventhAnswerSetProgrammingCompetitionDesign}
\setcitation{AspInPractice}{GKKS12AnswerSetSolvingPractice}
\setcitation{clasp}{GKS12Conflict-drivenanswersetsolvingFromtheory}
\setcitation{oclingo}{GGKOSS12StreamReasoningAnswerSetProgrammingPreliminary}
\setcitation{clingo}{GKKOSW16TheorySolvingMadeEasyClingo5}
\setcitation{gringo}{GST07GrinGoNewGrounderAnswerSetProgramming}
\setcitation{cmodels}{GLM04SAT-BasedAnswerSetProgramming}
\setcitation{DLV}{LPFEGPS06DLVsystemknowledgerepresentationreasoning}
\setcitation{GroundingBottleneck}{SPL14TwoApplicationsASP-PrologSystemDecomposablePrograms}
\setcitation{lazygroundingASP}{BW18ExploitingJustificationsLazyGroundingAnswerSet} 
\setcitation{justificationsAlpha}{BW18ExploitingJustificationsLazyGroundingAnswerSet}
\setcitation{ASP}{MT99StableModelsAlternativeLogicProgrammingParadigm} %{MT99StableModelsAlternativeLogicProgrammingParadigm,N99LogicProgramsStableModelSemanticsConstraint,L99AnswerSetPlanning} 

\setcitation{MaxSAT}{LM21MaxSATHardSoftConstraints}

\setcitation{KR}{vLP08HandbookKnowledgeRepresentation}

\setcitation{lazyclausegeneration}{OSC09Propagationlazyclausegeneration}
\setcitation{FP}{H89ConceptionEvolutionApplicationFunctionalProgrammingLanguages}
\setcitation{GroundingWithBounds}{WMD10GroundingFOFOIDBounds}
\setcitation{GroundWithBounds}{WMD10GroundingFOFOIDBounds}
\setcitation{LTC}{BJBDVD14SimulatingDynamicSystemsUsingLinearTime}
\setcitation{SPSAT}{DBDDM12SymmetryPropagationImprovedDynamicSymmetryBreaking}
\setcitation{BreakID}{DBBD16ImprovedStaticSymmetryBreakingSAT}
\setcitation{LCG}{S10LazyClauseGenerationCombiningPowerSAT}

\setcitation{GroundedFixpoints}{BVD15Groundedfixpointstheirapplicationsknowledgerepresentation}
\setcitation{PartialGroundedFixpoints}{BVD15PartialGroundedFixpoints}
\setcitation{LogicBlox}{GAK12LogicBloxPlatformLanguageTutorial}
\setcitation{proB}{LB08ProBautomatedanalysistoolsetBmethod}
\setcitation{NaturalInductions}{DV14Well-FoundedSemanticsPrincipleInductiveDefinitionRevisited} %TODO replace by journal publication if one is published
\setcitation{LP}{vK76SemanticsPredicateLogicProgrammingLanguage}

\setcitation{AF}{D95AcceptabilityArgumentsFundamentalRoleNonmonotonicReasoning}
\setcitation{ADF}{BW10AbstractDialecticalFrameworks}
\setcitation{ADFRevisited}{BSEWW13AbstractDialecticalFrameworksRevisited}
\setcitation{DefaultLogic}{R80LogicDefaultReasoning}
\setcitation{DL}{R80LogicDefaultReasoning}
\setcitation{AEL}{M85SemanticalConsiderationsNonmonotonicLogic}
\setcitation{minisat}{ES03ExtensibleSAT-solver}
\setcitation{completion}{C77NegationFailure}
\setcitation{ClarkCompletion}{C77NegationFailure}
\setcitation{wasp}{ADFLR13WASPNativeASPSolverBasedConstraint}
\setcitation{CEGAR}{CGJLV03Counterexample-guidedabstractionrefinementsymbolicmodelchecking}
\setcitation{CuttingPlane}{DFJ54SolutionLarge-ScaleTraveling-SalesmanProblem}
\setcitation{CuttingPlanes}{DFJ54SolutionLarge-ScaleTraveling-SalesmanProblem}
\setcitation{kodkod}{TJ07KodkodRelationalModelFinder}
\setcitation{CDCL}{SS99GRASPSearchAlgorithmPropositionalSatisfiability}
\setcitation{1UIP}{ZMMM01EfficientConflictDrivenLearningBooleanSatisfiability}
\setcitation{relevance}{JBDJD16RelevanceSATID}
\setcitation{relevance-implementation}{JDBJD16ImplementingRelevanceTrackerModule}
\setcitation{WFS}{VRS91Well-FoundedSemanticsGeneralLogicPrograms}
\setcitation{UnfoundedSet}{VRS91Well-FoundedSemanticsGeneralLogicPrograms}
\setcitation{UFS}{VRS91Well-FoundedSemanticsGeneralLogicPrograms}
\setcitation{stablesemantics}{GL88StableModelSemanticsLogicProgramming}
\setcitation{shatter}{ASM06EfficientSymmetryBreakingBooleanSatisfiability}
\setcitation{sbass}{DTW11Symmetry-breakinganswersetsolving}
\setcitation{lparsemanual}{url:lparsemanual}
\setcitation{AIC}{FGZ04Activeintegrityconstraints}
\setcitation{templates}{DvJD15Semanticstemplatescompositionalframeworkbuildinglogics}
\setcitation{templates2}{DvBJD16CompositionalTypedHigher-OrderLogicDefinitions}
\setcitation{sat-to-sat}{JTT16SAT-to-SATDeclarativeExtensionSATSolversNew}
\setcitation{sat-to-sat-qbf}{BJT16SolvingQBFInstancesNestedSATSolvers}
\setcitation{sat-to-sat-SO}{BJT16DeclarativeSolverDevelopmentCaseStudies}
\setcitation{XSB}{SW12XSBExtendingPrologTabledLogicProgramming}
\setcitation{KCmap}{DM02KnowledgeCompilationMap}
\setcitation{KnowledgeCompilationMap}{DM02KnowledgeCompilationMap}
\setcitation{TLA}{L02SpecifyingSystemsTLALanguageToolsHardware}
\setcitation{EventB}{A10ModelingEvent-B-SystemSoftwareEngineering}
\setcitation{MX}{MT05FrameworkRepresentingSolvingNPSearchProblems}
\setcitation{MIP}{S96Linearintegerprogramming-theorypractice}
\setcitation{PerfectModel}{P88DeclarativeSemanticsDeductiveDatabasesLogicPrograms}
\setcitation{SafeInductions}{BVD17SafeInductionsAlgebraicStudy}
\setcitation{alpha}{W17BlendingLazy-GroundingCDNLSearchAnswer-SetSolving}
\setcitation{omiga}{DEFWW12OMiGAOpenMindedGroundingOn-The-FlyAnswer}
\setcitation{gasp}{PDPR09GASPAnswerSetProgrammingLazyGrounding}
\setcitation{asperix}{LN09FirstVersionNewASPSolverASPeRiX}
\setcitation{CTL}{CE81DesignSynthesisSynchronizationSkeletonsUsingBranching-Time}
\setcitation{AFT-AIC}{BC18Fixpointsemanticsactiveintegrityconstraints}
\setcitation{UltimateApproximator}{DMT04Ultimateapproximationapplicationnonmonotonicknowledgerepresentation}
\setcitation{KripkeKleene}{F85Kripke-KleeneSemanticsLogicPrograms}
\setcitation{AFT-HO}{CRS18ApproximationFixpointTheoryWell-FoundedSemanticsHigher-Order} 
\setcitation{HereThere}{H30formalenRegelnintuitionistischenLogik}
\setcitation{dAEL}{VCBD16DistributedAutoepistemicLogicApplicationAccessControl}
\setcitation{SDD}{D11SDDNewCanonicalRepresentationPropositionalKnowledge}
\setcitation{HEX}{EIST06dlvhexProverSemantic-WebReasoningunderAnswer-Set}
\setcitation{wADF}{BSWW18WeightedAbstractDialecticalFrameworks}
\setcitation{wADFfix}{BPSWW18WeightedAbstractDialecticalFrameworksExtendedRevised}
\setcitation{TransitionSystems}{NOT06SolvingSATSATModuloTheoriesFrom}
\setcitation{galliwasp}{MG12GalliwaspGoal-DirectedAnswerSetSolver}
\setcitation{clingcon}{BKOS17Clingconnextgeneration}
\setcitation{lp2sat}{JN11CompactTranslationsNon-disjunctiveAnswerSetPrograms}
\setcitation{lp2mip}{LJN12AnswerSetProgrammingMixedIntegerProgramming}
\setcitation{lp2diff}{JNS09ComputingStableModelsReductionsDifferenceLogic}
\setcitation{lp2acyc}{GJR14AnswerSetProgrammingSATmoduloAcyclicity}
\setcitation{PB}{RM09Pseudo-BooleanCardinalityConstraints}
\setcitation{CuttingPlanes}{CCT87complexitycutting-planeproofs}
\setcitation{RoundingSAT}{EN18DivideConquerTowardsFasterPseudo-BooleanSolving}
\setcitation{PRS}{DG02InferenceMethodsPseudo-BooleanSatisfiabilitySolver}
\setcitation{sat4j}{LP10Sat4jlibraryrelease22}
\setcitation{mingo}{LJN12AnswerSetProgrammingMixedIntegerProgramming}
\setcitation{pbmodels}{LT05Pbmodels-SoftwareComputeStableModels}
\setcitation{HEF-LP}{GLL07Head-Elementary-Set-FreeLogicPrograms}
\setcitation{HCF-LP}{BD94PropositionalSemanticsDisjunctiveLogicPrograms}
\setcitation{aspcore2}{CFGIKKLM20ASP-Core-2InputLanguageFormat}
\setcitation{SemanticWeb}{AGvH12SemanticWebPrimer3rdEdition}
\setcitation{QBF}{KB21TheoryQuantifiedBooleanFormulas}
\setcitation{SMUS}{IPLM15SmallestMUSExtractionMinimalHittingSet}
\setcitation{CDCLsym}{MBCK18CDCLSymIntroducingEffectiveSymmetryBreakingSAT}
\setcitation{SEL}{DBB17SymmetricExplanationLearningEffectiveDynamicSymmetry}
\setcitation{Coq}{BC04InteractiveTheoremProvingProgramDevelopment-}
\setcitation{Isabelle}{NPW02IsabelleHOL-ProofAssistantHigher-OrderLogic}
\setcitation{HOL}{SN08BriefOverviewHOL4}
\setcitation{lean}{dU21Lean4TheoremProverProgrammingLanguage}
\setcitation{CakeML}{TMKFON19verifiedCakeMLcompilerbackend}
\setcitation{FRAT}{BCH22FlexibleProofFormatSATSolver-ElaboratorCommunication}
\setcitation{DRUP}{GN03VerificationProofsUnsatisfiabilityCNFFormulas}
\setcitation{RUP}{GN03VerificationProofsUnsatisfiabilityCNFFormulas}
\setcitation{GRIT}{CMS17EfficientCertifiedResolutionProofChecking}
\setcitation{TraceCheck}{url:TraceCheck}
\setcitation{LRAT}{CHHKS17EfficientCertifiedRATVerification}

\setcitation{PR}{HKB17ShortProofsWithoutNewVariables}
\setcitation{DRAT}{WHH14DRAT-trimEfficientCheckingTrimmingUsingExpressive}
\setcitation{satcomp2016}{BHJ17SATCompetition2016RecentDevelopments}
\setcitation{satcomp2013}{url:SATcomp2013}

\setcitation{Essence}{FHJHM08Essenceconstraintlanguagespecifyingcombinatorialproblems}
\setcitation{smodels}{SN01SmodelsSystem}
\setcitation{lparse}{S98ImplementationLocalGroundingLogicProgramsStable}
\setcitation{IDPZ3}{CVVD22IDP-Z3reasoningengineFO} %TODO replace by real publication
\setcitation{IDP-Z3}{CVVD22IDP-Z3reasoningengineFO} %} %TODO replace by real publication
\setcitation{Z3}{dB08Z3EfficientSMTSolver}

\setcitation{DMN}{DMN} %TODO is there something better than just this on-line thing?
\setcitation{Planning}{GNT04Automatedplanning-theorypractice}
\setcitation{AutomatedPlanning}{GNT04Automatedplanning-theorypractice}
\setcitation{RC2}{IMM19RC2EfficientMaxSATSolver}
\setcitation{enfragmo}{phd/Aavani14}
\setcitation{TPTP}{url:tptp}

\newtheorem{lemma}{Lemma}[section]
\newtheorem{proposition}{Proposition}[section]
\newtheorem{definition}{Definition}[section]
\newtheorem{theorem}{Theorem}[section]
\newtheorem{corollary}{Corollary}[section]

\newcommand{\lsem}{\mbox{$\lbrack\!\lbrack$}}
\newcommand{\rsem}{\mbox{$\rbrack\!\rbrack$}}

\newcommand{\mtrue}{\mathit{true}}
\newcommand{\mfalse}{\mathit{false}}
\newcommand{\mundef}{\mathit{undef}}

\newcommand*{\pnot}{\mathord{\sim}}

\newcommand{\mo}[1]{\llbracket#1\rrbracket}
\newcommand{\mos}[1]{\llbracket#1\rrbracket^{*}}

\newcommand{\mwrs}[3]{\llbracket#1\rrbracket_{#3}(#2)}
\newcommand{\mwrst}[3]{\llbracket#1\rrbracket^{*}_{#3}(#2)}

\newcommand{\aleq}[1][]{\leq_{#1}}
\newcommand{\HOL}{\ensuremath{\mathcal{HOL}}}

\newcommand{\TP}{T_{\mathsf{P}}}
\newcommand\AP{A_{\mathsf{P}}}
\newcommand{\ATP}{\AP}

\newcommand\lrule{\leftarrow}

\usepackage{xspace}
\newcommand\m[1]{\ensuremath{#1}\xspace}
\newcommand\basedom{\m{\iota}}
\newcommand\bool{\m{o}}
\newcommand{\cp}[1]{{#1}^c}

\newcommand\eg{e.g.,\xspace}
\newcommand\ie{i.e.,\xspace}

\usepackage[final]{listings}
\usepackage{caption}

\usepackage{amsmath}
\DeclareMathOperator\lfp{lfp}

\newcommand\newlineInListing{\mbox{\textcolor{red}{$\hookrightarrow$}\space}}

\lstdefinelanguage{holp}{
	otherkeywords={E',V'},
	morekeywords=[1]{X,Y,Z}, % First-order variables  (refer to objects)
	morekeywords=[2]{A,P,Q,R,V,V',E',E,S}, % Second-order variables  (refer to sets of objects)
	morekeywords=[3]{Ord,Prop}, % Higher-order variables
	morecomment=[l]{\%}
}
\newcommand\stickListings{\vspace{-10pt}}

\makeatletter
\renewenvironment{keywords}
  {\noindent\normalfont\small\rmfamily{\em \keywordsname}:}{\vspace{6.5\p@ \@plus 3\p@ \@minus 1\p@}\endtrivlist
   \vbox{\hrule \@width \hsize}}
\newcounter{Diff}
\newcounter{Start}
\def\ps@appendixheadings{\let\@mkboth\@gobbletwo
  \def\@oddhead{\setcounter{Diff}{\numexpr\value{page}-\value{Start}+2}\hfil{\itshape\@shorttitle}\hfil \llap{A-\theDiff}}\def\@evenhead{\setcounter{Diff}{\numexpr\value{page}-\value{Start}+2}\rlap{A-\theDiff}\hfil\itshape\@shortauthor\hfil}\let\@oddfoot\@empty
  \let\@evenfoot\@oddfoot
  \def\sectionmark##1{\markboth{##1}{}}\def\subsectionmark##1{\markright{##1}}}
\makeatother

\begin{document}

\lefttitle{Bogaerts et al.}

\jnlPage{1}{14}
\jnlDoiYr{2024}
\doival{10.1017/xxxxx}

\title[The Stable Model Semantics for Higher-Order Logic Programming]{%
  The Stable Model Semantics for\\ Higher-Order Logic Programming%
  \thanks{This work was partially supported by Fonds Wetenschappelijk Onderzoek -- Vlaanderen (project G0B2221N)
          and by a research project which is implemented in the framework of H.F.R.I call ``Basic research
          Financing (Horizontal support of all Sciences)'' under the National Recovery and Resilience
          Plan ``Greece 2.0'' funded by the European Union - NextGenerationEU (H.F.R.I. Project Number: 16116).}}

\begin{authgrp}
\author{%
  \sn{Bogaerts} \gn{Bart}$^{\dagger\ddagger}$,
  \sn{Charalambidis} \gn{Angelos}$^\circledast$,
  \sn{Chatziagapis} \gn{Giannos}$^\divideontimes$,
  \sn{Kostopoulos} \gn{Babis}$^{\circledast}$,
  \sn{Pollaci} \gn{Samuele}$^{\dagger\ddagger}$ \and
  \sn{Rondogiannis} \gn{Panos}$^\divideontimes$}
\affiliation{$\dagger$Vrije Universiteit Brussel, Belgium}
\affiliation{$\ddagger$Katholieke Universiteit Leuven, Belgium}
\affiliation{$\circledast$Harokopio University of Athens, Greece}
\affiliation{$\divideontimes$National and Kapodistrian University of Athens, Greece}
\end{authgrp}

\history{\sub{xx xx xxxx;} \rev{xx xx xxxx;} \acc{xx xx xxxx}}

\maketitle

\begin{abstract}
We propose a stable model semantics for higher-order logic programs. Our semantics is developed
using Approximation Fixpoint Theory (AFT), a powerful formalism that has successfully been used to give
meaning to diverse non-monotonic formalisms. The proposed semantics generalizes the classical two-valued
stable model semantics of~\citep{GL88StableModelSemanticsLogicProgramming} as-well-as the three-valued
one of~\citep{P90Well-FoundedSemanticsCoincidesThree-ValuedStableSemantics}, retaining their
desirable properties. Due to the use of AFT, we also get for free alternative semantics for higher-order
logic programs, namely \emph{supported model}, \emph{Kripke-Kleene}, and \emph{well-founded}. Additionally,
we define a broad class of \emph{stratified higher-order logic programs} and demonstrate that they have
a unique two-valued higher-order stable model which coincides with the well-founded semantics of such programs.
We provide a number of examples in different application domains, which demonstrate that higher-order logic
programming under the stable model semantics is a powerful and versatile formalism, which can potentially
form the basis of novel ASP systems.
\ifarxiv
This work is under consideration for acceptance in TPLP.%
\fi
\end{abstract}


\section{Introduction}
Recent
research~\citep{CHRW13ExtensionalHigher-OrderLogicProgramming,CRS18ApproximationFixpointTheoryWell-FoundedSemanticsHigher-Order,CRT18Higher-orderlogicprogrammingexpressivelanguagerepresenting}
has demonstrated that it is possible to design higher-order logic programming
languages that have powerful expressive capabilities and simple and elegant
semantic properties. These languages are genuine extensions of classical
(first-order) logic programming: for example,
\citet{CHRW13ExtensionalHigher-OrderLogicProgramming} showed that positive
higher-order logic programs have a Herbrand model intersection property and this
least Herbrand model can also be produced as the least fixpoint of a continuous
immediate consequence operator. In other words, crucial semantic results of
classical (positive) logic programs transfer directly to the higher-order
setting.

The above positive results, created the hope and expectation that all major
achievements of first-order logic programming could transfer to the higher-order
world. Despite this hope, it was not clear until now whether it is possible to
define a \emph{stable model semantics} for higher-order logic programs that would
generalize the seminal work of~\citet{GL88StableModelSemanticsLogicProgramming}.
For many extensions of standard logic programming, it is possible to generalize
the \emph{reduct} construction
of~\citeauthor{GL88StableModelSemanticsLogicProgramming} to obtain a stable
model semantics, as illustrated for instance by
\citet{FPL11Semanticscomplexityrecursiveaggregatesanswerset} for an extension of
logic programs with aggregates. For higher-order programs, however, it is not
clear whether a reduct-based definition makes sense. The most important reason
why it is challenging to define a higher-order reduct, is that using the
powerful abstraction mechanisms that higher-order languages provide, one can
\emph{define} negation inside the language, for instance by the rule
\lstinline|neg X  :- ~ X| and use \lstinline|neg| everywhere in the program
where otherwise negation would be used, rendering syntactic definitions based on
occurrences of negation difficult to apply.

Apart from scientific curiosity, the definition of a stable model semantics for
higher-order logic programs also serves solid practical goals: there has been a
quest for extending the power of ASP
systems~\citep{BJT16Stable-unstablesemanticsBeyondNPnormallogic,ART19BeyondNPQuantifyingoverAnswerSets,FLRSS21PlanningIncompleteInformationQuantifiedAnswerSet},
and higher-order logic programming under the stable model semantics may prove to
be a promising solution.

In this paper we define a stable model semantics for higher-order logic
programs. Our semantics is developed using Approximation Fixpoint Theory
(AFT)~\citep{DMT04Ultimateapproximationapplicationnonmonotonicknowledgerepresentation}.
AFT is a powerful lattice-theoretic formalism that was originally developed to
unify semantics of logic programming, autoepistemic logic (AEL) and default
logic (DL) and was used to resolve a long-standing open question about the
relation between AEL and DL semantics
\citep{DMT03Uniformsemantictreatmentdefaultautoepistemiclogics}. Afterwards, it
has been applied to several other fields, including abstract argumentation
\citep{S13Approximatingoperatorssemanticsabstractdialecticalframeworks},  active
integrity constraints \citep{BC18Fixpointsemanticsactiveintegrityconstraints},
stream reasoning \citep{A20Fixedpointsemanticsstreamreasoning}, and constraint
languages for the semantic web \citep{BJ21FixpointSemanticsRecursiveSHACL}. In
these domains, AFT has been used to define new semantics without having to
reinvent the wheel (for instance, if one uses AFT to define a stable semantics,
well-known properties such as minimality results will be automatic), to study
the relation to other formalisms, and even to discover bugs in the original
semantics \citep{B19WeightedAbstractDialecticalFrameworksthroughLens}. To apply
AFT to a new domain, what we need to do is define a suitable semantic operator
on a suitable set of ``partial interpretations''. Once this operator is
identified, a family of well-known semantics and properties immediately rolls
out of the abstract theory. In this paper, we construct such an operator for
higher-order logic programs. Since our operator coincides with
Fitting's~(\citeyear{F02Fixpointsemanticslogicprogrammingsurvey})  three-valued
immediate consequence operator for the case of standard logic programs, we
immediately know that our resulting stable semantics generalizes the classical
two-valued stable model semantics
of~\citet{GL88StableModelSemanticsLogicProgramming} as well as the three-valued
one of~\citet{P90Well-FoundedSemanticsCoincidesThree-ValuedStableSemantics}.

The main idea of our construction is to interpret the higher-order predicates of
our language as three-valued relations over two-valued objects, \ie as functions
that take classical relations as arguments and return $\mtrue$, $\mfalse$, or
$\mundef$. We demonstrate that such relations are equivalent to appropriate
pairs of (classical) two-valued relations. The pair-representation gives us the
basis to apply AFT, and to obtain, in a simple and transparent manner, the
stable model semantics. At the same time, thanks to the versatility of AFT,
without any additional effort, we obtain several alternative semantics for
higher-order logic programs, namely \emph{supported model},
\emph{Kripke-Kleene}, and \emph{well-founded} semantics. In particular, we argue
that our well-founded semantics remedies certain deficiencies that have been
observed in other attempts to define such a semantics for higher-order
formalisms~\citep{DvJD15Semanticstemplatescompositionalframeworkbuildinglogics,DvBJD16CompositionalTypedHigher-OrderLogicDefinitions,CRS18ApproximationFixpointTheoryWell-FoundedSemanticsHigher-Order}.
We study properties of our novel semantics and to do so, we define a broad class
of \emph{stratified higher-order logic programs}. This is a non-trivial task
mainly due to the fact that in the higher-order setting non-monotonicity can be
well-hidden
\citep{RS17intricaciesthree-valuedextensionalsemanticshigher-orderlogic} and
stratification will hence have to take more than just occurrences of negation
into account. We demonstrate that stratified programs, as expected, indeed have
a unique two-valued higher-order stable model, which coincides with the
well-founded model of such programs. We feel that these results create a solid
and broad foundation for the semantics of higher-order logic programs with
negation. Finally, from a practical perspective, we showcase our semantics on
three different examples.
In Section~\ref{sec:motivating}, we start with max-clique, a simple
graph-theoretic problem  which we use to familiarize the reader with our
notation and to demonstrate the power of abstraction. In
Section~\ref{sec:applications}, we study more intricate applications, namely
semantics for abstract argumentation and Generalized Geography, which is a
$\mathsf{PSPACE}$-complete problem. These examples illustrate that higher-order
logic programming under the stable model semantics is a powerful and versatile
formalism, which can potentially form the basis of novel ASP systems.

% \ifarxiv
% The proofs of all main results of the paper have been provided in an appendix as
% supplementary material.
% \fi

\section{A Motivating Example}\label{sec:motivating}
In this section, we illustrate our higher-order logic programming language on
the max-clique problem. A complete solution is included in
Listing~\ref{list:maxclique}. We will assume an undirected graph is given by
means of a unary predicate \lstinline|v| (containing all the nodes of the graph)
and a binary predicate \lstinline|e| representing the edge-relation (which we
assume to be symmetric). Lines~\ref{line:choice1} and~\ref{line:choice2} contain
the standard trick that exploits an even loop of negation for simulating a
choice, which in modern ASP input formats
\citep{CFGIKKLM20ASP-Core-2InputLanguageFormat} would be abbreviated by a choice
rule construct \lstinline|{pick X : v X}|. Line~\ref{line:clique} defines what
it means to be a clique. In this line the (red) variable \lstinline|P| is a
first-order variable; it ranges over all \emph{sets} of domain elements, whereas
(blue) zero-order variables such as \lstinline|X| in Line~\ref{line:choice1}
range over actual elements of the domain. A set of elements is a clique if it
\emph{(i)} is a subset of \lstinline|v|, and \emph{(ii)} contains no two nodes
without an edge between them. Failures to satisfy the second condition are
captured by the predicate \lstinline|hasNonEdge|. The predicate
\lstinline|clique| is a second-order predicate. Formally, we will say its type
is $(\basedom\to\bool)\to\bool$: it takes as input a  relation of type
$\basedom\to\bool$, \ie a set of base domain elements of type \basedom  and it
returns a Boolean (type $\bool$). In other words, the interpretation of
\lstinline|clique| will be a set of sets. Next, line~\ref{line:maxclique}
defines the second-order predicate \lstinline|maxclique|, which is true
precisely for those sets \lstinline|P| that are subset-maximal among the set of
all cliques, and line~\ref{line:assert} asserts, using the standard trick with
an odd loop over negation that \lstinline|pick| must indeed be in
\lstinline|maxclique|.

This definition of \lstinline|maxclique| makes use of a third-order predicate
\lstinline|maximal| which works with an arbitrary binary relation for comparing
sets (here: the subset relation), as well as an arbitrary unary predicate over
sets (here: the \lstinline|clique| predicate). Listing~\ref{list:generic}
provides definitions of \lstinline|maximal|, \lstinline|equal|, and other
generic predicates. Note that equality between predicates is not a primitive of
the language: we define it in Line~\ref{line:equal} of
Listing~\ref{list:generic}. On the other hand, equality between atomic objects
(\lstinline|=|), which we use in Line~\ref{line:nonedge} of
Listing~\ref{list:maxclique}, is a primitive of the language. These generic
definitions, which can be reused in different applications, illustrate the power
of higher-order modelling: it enables reuse and provides great flexibility, \eg
if we are interested in cardinality-maximal cliques, we only need to replace
subset by an appropriate relation comparing the size of two predicates. Also
note that our solution has only a single definition of what it means to be a
clique. This definition is used both to state that \lstinline|pick| is a clique
(the first atom of the rule defining \lstinline|maximal| guarantees this) and to
check that there are no larger cliques (in the rule defining
\lstinline|nonmaximal|).

\begin{lstlisting}[caption=Max-clique problem using stable semantics for higher-order logic programs.,
  label=list:maxclique]
% Pick a set of vertices (emulate choice rule)
pick X  :- v X, ~(npick X). |\label{line:choice1}|
npick X :- v X, ~(pick X).  |\label{line:choice2}|
% Define what it means for a set of vertices to be a clique
hasNonEdge P :- P X, P Y, ~(X=Y), ~(e X Y). |\label{line:nonedge}|
clique P :- subset P v,  ~(hasNonEdge P).  |\label{line:clique}|
% Define what it means to be a max-clique:
maxclique P :- maximal subset clique P.  |\label{line:maxclique}|
% The selected set should be a max-clique
f :- ~f, ~(maxclique pick).  |\label{line:assert}|
\end{lstlisting}
\stickListings
\begin{lstlisting}[caption=Definitions of generic higher-order predicates.,
  label=list:generic]
% Define generic higher-order predicates: subset, equal, maximal
nonsubset P Q :- P X, ~(Q X).
subset P Q :-  ~(nonsubset P Q).
equal P Q :- subset P Q, subset Q P. |\label{line:equal}|
% maximal Ord Prop  P means: P is Ord-maximal among sets satisfying Prop
maximal Ord Prop P    :-  Prop P, ~(nonmaximal Ord Prop P).
nonmaximal Ord Prop P :-  Prop Q, Ord P Q, ~(equal P Q).
\end{lstlisting}

\section{$\HOL$: A Higher-Order Logic Programming Language}\label{sec:hol}
In this section we define the syntax of the language $\HOL$ that we use
throughout the paper. For simplicity reasons, the syntax of $\HOL$ does not
include function symbols; this is a restriction that can easily be lifted.
$\HOL$ is based on a simple type system with two base types: $\bool$, the
Boolean domain, and $\basedom$, the domain of data objects. The composite types
are partitioned into \emph{predicate} ones (assigned to predicate symbols) and
\emph{argument} ones (assigned to parameters of predicates).
\begin{definition}
Types are either \emph{predicate} or \emph{argument}, denoted by $\pi$
and $\rho$ respectively, and defined as:
\begin{align*}
%\sigma & := \iota \mid (\iota \rightarrow \sigma) \\
\pi  & := \bool \mid (\rho \to \pi) \\
\rho & := \basedom \mid \pi
\end{align*}
\end{definition}

As usual, the binary operator $\to$ is right-associative. It can be easily seen
that every predicate type $\pi$ can be written in the form
$\rho_1 \to \cdots \to \rho_n \rightarrow \bool$, $n\geq 0$ (for $n=0$ we assume that $\pi=\bool$).
We proceed by defining the syntax of $\HOL$.

\begin{definition}
The \emph{alphabet} of $\HOL$ consists of the following: \emph{predicate variables}
of every predicate type $\pi$ (denoted by capital letters such as $\mathsf{P,Q,\ldots})$;
\emph{predicate constants} of every predicate type $\pi$ (denoted by lowercase letters such as $\mathsf{p,q,\ldots}$);
\emph{individual variables} of type $\basedom$ (denoted by capital letters such as $\mathsf{X,Y,\ldots}$);
\emph{individual constants} of type $\basedom$ (denoted by lowercase letters such as $\mathsf{a,b,\ldots}$);
the \emph{equality}  constant $\approx$ of type $\basedom \to \basedom \to o$ for comparing individuals of type $\basedom$;
the \emph{conjunction} constant $\wedge$ of type $\bool \to \bool \to \bool$;
the \emph{rule operator} constant $\lrule$ of type $\bool \to \bool \to \bool$; and
the \emph{negation} constant $\pnot$ of type $\bool \to \bool$.
\end{definition}

Arbitrary variables (either predicate or individual ones) will usually be denoted by $\mathsf{R}$. % and its subscripted versions.% Save a line

\begin{definition}
The \emph{terms} and \emph{expressions} of $\HOL$ are defined as follows.
Every predicate variable/constant and every individual variable/constant is a
term; if $\mathsf{E}_1$ is a term of type $\rho \to \pi$ and $\mathsf{E}_2$ a
term of type $\rho$ then $(\mathsf{E}_1\ \mathsf{E}_2)$ is a term of type $\pi$.
\label{def:expressions}
Every term is also an expression;
if $\mathsf{E}$ is a term of type $\bool$ then $(\pnot \mathsf{E})$ is an expression of type $\bool$;
if $\mathsf{E}_1$ and $\mathsf{E}_2$ are terms of type $\basedom$, then $(\mathsf{E}_1\approx \mathsf{E}_2)$ is an expression of type $\bool$.
\end{definition}

We will omit parentheses when no confusion arises. To denote that an expression
$\mathsf{E}$ has type $\rho$ we will often write $\mathsf{E}:\rho$.

\begin{definition}
A \emph{rule} of $\HOL$ is a formula
$\mathsf{p}\ \mathsf{R}_1 \cdots \mathsf{R}_n \lrule \mathsf{E}_1 \land \ldots \land \mathsf{E}_m$,
where $\mathsf{p}$ is a predicate constant of type $\rho_1 \to \cdots \to \rho_n \to \bool$,
$\mathsf{R}_1,\ldots,\mathsf{R}_n$ are distinct variables of types $\rho_1,\ldots,\rho_n$ respectively and
the $\mathsf{E}_i$ are expressions of type $\bool$.
The term $\mathsf{p}\ \mathsf{R}_1 \cdots \mathsf{R}_n$ is the \emph{head} of the rule and
$ \mathsf{E}_1 \land \ldots \land \mathsf{E}_m$ is the \emph{body} of the rule.
A \emph{program} $\mathsf{P}$ of $\HOL$ is a finite set of rules.
\end{definition}

We will often follow the common logic programming notation and write
$\mathsf{E}_1,\ldots,\mathsf{E}_m$ instead of
$\mathsf{E}_1 \wedge \cdots \wedge \mathsf{E}_m$ for the body of a rule.
For brevity reasons, we will often denote
a rule as $\mathsf{p} \ \overline{\mathsf{R}} \lrule \mathsf{B}$, where
$\overline{\mathsf{R}}$ is a shorthand for a sequence of variables
$\mathsf{R}_1 \cdots \mathsf{R}_n$ and $\mathsf{B}$ represents a conjunction of expressions of
type $\bool$.

\section{The Two-Valued Semantics of $\HOL$}\label{sec:two-valued-semantics}
In this section we define an immediate consequence operator for $\HOL$ programs,
which is an extension of the classical $\TP$ operator for first-order logic
programs. We start with the semantics of the types of our language. In the
following, we denote by $U_{\mathsf{P}}$ the \emph{Herbrand universe} of
$\mathsf{P}$, namely the set of all constants of the program.

The semantics of the base type $\bool$ is the classical Boolean domain $\{\mathit{true}, \mathit{false}\}$
and that of the base type $\basedom$ is $U_{\mathsf{P}}$.
The semantics of types of the form $\rho \to \pi$ is the set of all functions
from the domain of type $\rho$ to that of type $\pi$.
We define, simultaneously
with the meaning of every type, a partial order on the elements of the type.

\begin{definition}\label{def:orders}
Let $\mathsf{P}$ be an $\HOL$ program. We define the (two-valued) meaning of a type with respect to $U_{\mathsf{P}}$, as follows:
\begin{itemize}
  \item $\mo{\bool}_{U_{\mathsf{P}}} = \{\mathit{true}, \mathit{false}\}$.
        The partial order $\leq_\bool$ is the usual one induced by the ordering  $\mfalse <_\bool \mtrue$
  \item $\mo{\basedom}_{U_{\mathsf{P}}} = U_{\mathsf{P}}$.
        The partial order $\leq_\basedom$ is the
        trivial one defined as $d \leq_\basedom d$ for all $d \in U_{\mathsf{P}}$
  \item $\mo{\rho \to \pi}_{U_{\mathsf{P}}} = \mo{\rho}_{U_{\mathsf{P}}} \to \mo{\pi}_{U_{\mathsf{P}}}$,
        namely the set of all functions from  $\mo{\rho}_{U_{\mathsf{P}}}$ to $\mo{\pi}_{U_{\mathsf{P}}}$.
        The partial order $\leq_{\rho \to \pi}$ is defined as: for all $f,g \in \mo{\rho \to \pi}_{U_{\mathsf{P}}}$,
        $f \leq_{\rho \to \pi} g$ iff $f(d) \leq_{\pi} g(d)$ for all $d \in \mo{\rho}_{U_{\mathsf{P}}}$.
\end{itemize}
\end{definition}
The subscripts from the above partial orders will be omitted when they are obvious from context.
Moreover, we will omit the subscript $U_{\mathsf{P}}$ assuming that our semantics is defined with
respect to a specific program $\mathsf{P}$.

As we mentioned before, each predicate type $\pi$ can be written in the form $\rho_1\to\dots\to\rho_n\to \bool$. Elements of $\mo{\pi}$ can be thought of, alternatively, as subsets of $\mo{\rho_1}\times\dots\times\mo{\rho_n}$ (the set contains precisely those $n$-tuples mapped to $\mtrue$).
Under this identification, it can be seen that $\leq_\pi$ simply becomes the subset relation.

\begin{proposition}\label{semantics_of_types_lattice}
For every predicate type $\pi$, $(\mo{\pi}, \leq_\pi)$ is a complete lattice.
\end{proposition}

In the following, we denote by $\bigvee_{\leq_{\pi}}$ and
$\bigwedge_{\leq_{\pi}}$ the corresponding lub and glb operations of the above
lattice. When viewing elements of $\pi$ as \emph{sets}, $\bigvee_{\leq_{\pi}}$
is just the union operator and $\bigwedge_{\leq_{\pi}}$ the intersection. We now
proceed to define Herbrand interpretations and states.
\begin{definition}\label{def:interpretation_Herbrand} A \emph{Herbrand
interpretation $I$ of a program $\mathsf{P}$} assigns to each individual
constant $\mathsf{c}$ of $\mathsf{P}$, the element $I(\mathsf{c}) = \mathsf{c}$,
and to each predicate constant $\mathsf{p} : \pi$ of $\mathsf{P}$, an element
$I(\mathsf{p}) \in \mo{\pi}$.
\end{definition}
We will denote the set of Herbrand interpretations of a program $\mathsf{P}$
with $H_\mathsf{P}$. We define a partial order on $H_\mathsf{P}$ as
follows: for all $I, J \in H_\mathsf{P}$, $I \leq J$
iff for every predicate constant $\mathsf{p} : \pi$ that appears in
$\mathsf{P}$, $I(\mathsf{p}) \aleq[\pi] J(\mathsf{p})$.
The following proposition demonstrates that the space of interpretations is a
complete lattice. This is an easy consequence of
Proposition~\ref{semantics_of_types_lattice}. % and omitted.
\begin{proposition}\label{interp-states-lattices}
Let $\mathsf{P}$ be a program. Then, $(H_{\mathsf{P}}, \leq)$ is a complete lattice.
\end{proposition}
\begin{definition}\label{def:state_Herbrand}
	A \emph{Herbrand state} $s$ of a program $\mathsf{P}$ is a function that
	assigns to each argument variable $\mathsf{R}$ of type $\rho$, an element
	$s(\mathsf{R}) \in \mo{\rho}$.  We denote the set of Herbrand states with
	$S_\mathsf{P}$.
\end{definition}
In the following, $s[\mathsf{R}_1/d_1,\ldots,\mathsf{R}_n/d_n]$ is used to denote a state that is identical to $s$ the
only difference being that the new state assigns to each $\mathsf{R}_i$ the corresponding value $d_i$;
for brevity, we will also denote it by $s[\overline{\mathsf{R}}/\overline{d}]$.

We proceed to define the (two-valued) semantics of $\cal{HOL}$ expressions and bodies.
\begin{definition}\label{standard-semantics}
Let $\mathsf{P}$ be a program, $I$ a Herbrand interpretation of $\mathsf{P}$, and
$s$ a Herbrand state. Then, the semantics of expressions and bodies is
defined as follows:
\begin{enumerate}
  \item $\mwrs{\mathsf{R}}{I}{s} = s(\mathsf{R})$
  \item $\mwrs{\mathsf{c}}{I}{s} = I(\mathsf{c}) = \mathsf{c}$
  \item $\mwrs{\mathsf{p}}{I}{s} = I(\mathsf{p})$
  \item $\mwrs{(\mathsf{E}_1\ \mathsf{E}_2)}{I}{s} = \mwrs{\mathsf{E}_1}{I}{s}\ \mwrs{\mathsf{E}_2}{I}{s}$
  \item $\mwrs{(\mathsf{E}_1\approx \mathsf{E}_2)}{I}{s} = \begin{cases}
    \mathit{true},  & \text{if } \mwrs{\mathsf{E}_1}{I}{s} = \mwrs{\mathsf{E}_2}{I}{s} \\
    \mathit{false}, & \text{otherwise}
    \end{cases}$
  \item $\mwrs{(\sim \mathsf{E})}{I}{s} = \begin{cases}
    \mathit{true},  & \text{if } \mwrs{\mathsf{E}}{I}{s} = \mathit{false} \\
    \mathit{false}, & \text{otherwise}
    \end{cases}$
  \item $\mwrs{(\mathsf{E}_1 \wedge \cdots \wedge \mathsf{E}_m)}{I}{s} =
    \bigwedge_{\leq_\bool}\{\mwrs{\mathsf{E}_1}{I}{s},\ldots,\mwrs{\mathsf{E}_m}{I}{s}\}$
\end{enumerate}
\end{definition}
We can now formally define the notion of \emph{model} for $\HOL$ programs.
\begin{definition}\label{def:twovalued-models}
Let $\mathsf{P}$ be a program and $M$ be a two-valued Herbrand interpretation of $\mathsf{P}$.
Then, $M$ is a \emph{two-valued Herbrand model} of $\mathsf{P}$ iff for every rule
$\mathsf{p}\ \overline{\mathsf{R}} \lrule \mathsf{B}$ in $\mathsf{P}$
and for every Herbrand state $s$,
$\mwrs{\mathsf{B}}{M}{s} \leq_o \mwrs{\mathsf{p}\ \overline{\mathsf{R}}}{M}{s}$.
\end{definition}
Since we have a mechanism to evaluate bodies of rules, we can define the \emph{immediate consequence operator} for
$\HOL$ programs, which generalizes the corresponding operator for classical (first-order)
logic programs of \citep{vK76SemanticsPredicateLogicProgrammingLanguage}.
\begin{definition}
  Let $\mathsf{P}$ be a program. The mapping $T_\mathsf{P} : H_\mathsf{P} \to H_\mathsf{P}$
  is called the \emph{immediate consequence operator for $\mathsf{P}$} and is defined for every predicate constant
  $\mathsf{p} : \rho_1 \to \cdots \to \rho_n \to \bool$ and all $d_1 \in \mo{\rho_1},\ldots,d_n \in \mo{\rho_n}$, as:
  $
  \TP(I)(\mathsf{p})\ \overline{d} =
              \bigvee\nolimits_{\leq_\bool}\{
                  \mwrs{\mathsf{B}}{I}{s[\overline{\mathsf{R}}/\overline{d}]}  \mid
                  \mbox{$s\in S_{\mathsf{P}}$ and $(\mathsf{p}\ \overline{\mathsf{R}} \leftarrow \mathsf{B})$
                  in $\mathsf{P}$}\}$.
\end{definition}

Since a program may contain negation, $T_\mathsf{P}$ is not necessarily monotone. %, since our programs may contain negation.
In fact, perhaps somewhat surprisingly, $T_\mathsf{P}$ can even be non-monotone for negation-free programs
such as \lstinline|p  :- r(p)|, where \lstinline|p| is of type $o$ and \lstinline|r| is a predicate constant
of type $o\to o$.\footnote{To see the non-monotonicity of $T_{\mathsf{P}}$ for this program, consider an
interpretation $I_0$ which assigns to \lstinline|p| the value $\mfalse$ and to \lstinline|r| the negation operation
$\mathit{neg}:\bool\to\bool: \mtrue\mapsto\mfalse,\mfalse\mapsto\mtrue$. Consider also an interpretation $I_1$ which
is identical to $I_0$ the only difference being that it assigns to \lstinline|p| the value $\mtrue$. It can be verified that $I_0\leq I_1$ but $T_{\mathsf{P}}(I_0)\not\leq T_{\mathsf{P}}(I_1)$.}

As expected, $T_\mathsf{P}$ characterizes the models of $\mathsf{P}$, as the following proposition suggests.
\begin{proposition}\label{model-iff-tp-prefixpoint}
Let $\mathsf{P}$ be a program and $I \in H_\mathsf{P}$. Then, $I$ is a model of $\mathsf{P}$
iff $I$ is a pre-fixpoint of $T_\mathsf{P}$ (i.e., $\TP(I)\leq I$).
\end{proposition}

\section{The Three-Valued Semantics of $\HOL$}\label{sec:three-valued-semantics}
In this section we define an alternative semantics for $\HOL$ types and
expressions, based on a three-valued truth space. As in first-order logic
programming, the purpose of the third truth value is to assign meaning to programs
that contain circularities through negation. Since we are dealing with higher-order
logic programs, we must define three-valued relations at all orders of the type
hierarchy. These three-valued relations are functions that take two-valued arguments
and return a three-valued truth result.

Due to the three-valuedness of our base domain $o$, all our domains inherit two
distinct ordering relations, namely $\leq$ (the \emph{truth ordering}) and
$\preceq$ (the \emph{precision ordering}).
\begin{definition}\label{def:orders_three-valued}
Let $\mathsf{P}$ be a program. We define the (three-valued) meaning of a type with respect to $U_{\mathsf{P}}$, as follows:
\begin{itemize}
\item $\mos{\bool}_{U_{\mathsf{P}}} = \{\mfalse,\mundef,\mtrue\}$.
      The partial order $\leq_\bool$ is the one induced by the ordering  $\mfalse <_\bool \mundef <_\bool \mtrue$;
      the partial order $\preceq_\bool$ is the one induced by the ordering $\mundef \prec_\bool \mfalse$ and $\mundef \prec_\bool \mtrue$.

\item $\mos{\basedom}_{U_{\mathsf{P}}} = U_{\mathsf{P}}$.
      The partial order $\leq_\basedom$ is defined as $d \leq_\basedom d$ for all $d \in U_{\mathsf{P}}$.
      The partial order $\preceq_\basedom$ is also defined as $d \preceq_\basedom d$ for all $d\in U_{\mathsf{P}}$.

\item $\mos{\rho \to \pi}_{U_{\mathsf{P}}} = \mo{\rho}_{U_{\mathsf{P}}} \to \mos{\pi}_{U_{\mathsf{P}}}$.
      The partial order $\leq_{\rho \to \pi}$ is defined as follows:
      for all $f,g \in \mos{\rho \to \pi}_{U_{\mathsf{P}}}$,
      $f \leq_{\rho \to \pi} g$ iff $f(d) \leq_{\pi} g(d)$ for all $d \in \mo{\rho}_{U_{\mathsf{P}}}$.
      The partial order $\preceq_{\rho \to \pi}$ is defined as follows:
      for all $f,g \in \mos{\rho \to \pi}_{U_{\mathsf{P}}}$,
      $f \preceq_{\rho \to \pi} g$ iff $f(d) \preceq_{\pi} g(d)$ for all $d \in \mo{\rho}_{U_{\mathsf{P}}}$.
\end{itemize}
\end{definition}

We omit subscripts when unnecessary.
It can be easily verified that for every $\rho$ it holds $\mo{\rho} \subseteq \mos{\rho}$.
In other words, every two-valued element is also a three-valued one.
Moreover, the $\leq_\rho$ ordering in the above definition is an
extension of the $\leq_\rho$ ordering in Definition~\ref{def:orders}.

% The following propositions are needed to ensure the well-definedness of the definitions that follow.
%
\begin{proposition}\label{semantics_of_types_lattice_cpo}
For every predicate type $\pi$, $(\mos{\pi}, \leq_\pi)$
is a complete lattice and $(\mos{\pi}, \preceq_\pi)$ is a complete meet-semilattice (\ie
every non-empty subset of $\mos{\pi}$ has a $\preceq_\pi$-greatest lower bound).
\end{proposition}

We denote by $\bigvee_{\leq_{\pi}}$ and $\bigwedge_{\leq_{\pi}}$ the
 lub and glb operations of the lattice $(\mos{\pi}, \leq_{\pi})$;
it can easily be verified that these operations are extensions of the
corresponding operations implied by
Proposition~\ref{semantics_of_types_lattice}. We denote by
$\bigwedge_{\preceq_{\pi}}$ the glb in $(\mos{\pi}, \preceq_{\pi})$. Just like
how a two-valued interpretation of a predicate $\pi$ of type
$\rho_1\to\dots\to\rho_n\to \bool$ can be viewed as a \emph{set}, an element of
$\mos{\pi}$ can be viewed as a \emph{partial set}, assigning to each tuple in
$\mo{\rho_1}\times \dots \times \mo{\rho_n}$ one of three truth values
($\mtrue$, meaning the tuple is \emph{in} the set, $\mfalse$ meaning it is
not in the set, or $\mundef$ meaning it is not determined if it is in the set or
not). This explains why the \emph{arguments} are interpreted classically: a
partial set decides for each \emph{actual} (\ie two-valued) object whether it is
in the set or not; it does not make statements about \emph{partial} (\ie three-valued) objects.
Due to the fact that the arguments of relations are interpreted classically, the definition of Herbrand
states that we use below, is the same as that of Definition~\ref{def:state_Herbrand}.
A \emph{three-valued Herbrand interpretation} is defined analogously to a
two-valued one (Definition~\ref{def:interpretation_Herbrand}), the only difference being that
the meaning of a predicate constant $\mathsf{p} : \pi$ is now an element of $\mos{\pi}_{U_{\mathsf{P}}}$.
We will use caligraphic fonts (\eg ${\cal I},{\cal J}$) to differentiate three-valued interpretations
from two-valued ones. The set of all three-valued Herbrand interpretations is denoted by ${\cal H}_\mathsf{P}$.
Since $\mo{\pi} \subseteq \mos{\pi}$ it also follows that $H_\mathsf{P} \subseteq {\cal H}_\mathsf{P}$.
\begin{definition}\label{ordering-interpretations}
Let $\mathsf{P}$ be a program. We define the partial orders $\leq$ and $\preceq$ on ${\cal H}_\mathsf{P}$
as follows: for all ${\cal I}, {\cal J} \in  {\cal H}_\mathsf{P}$, ${\cal I} \leq {\cal J}$ (respectively,
${\cal I} \preceq {\cal J}$) iff for every predicate type $\pi$ and for every predicate constant
$\mathsf{p} : \pi$ of $\mathsf{P}$, ${\cal I}(\mathsf{p}) \leq_\pi {\cal J}(\mathsf{p})$ (respectively,
${\cal I}(\mathsf{p}) \preceq_\pi {\cal J}(\mathsf{p})$). %\samuele{We should add that ${\cal I}(\mathsf{c})\leq_\iota {\cal J}(\mathsf{c})$ for every individual constant $\mathsf{c}$ of type $\iota$, i.e.\ they coincide on individual constants.}
\end{definition}

\begin{definition}\label{tuple-semantics}
  Let $\mathsf{P}$ be a program, ${\cal I}$ a three-valued Herbrand interpretation of $\mathsf{P}$, and
  $s$ a Herbrand state. The \emph{three-valued semantics} of expressions and bodies
  is defined as follows:
\begin{enumerate}
  \item $\mwrst{\mathsf{R}}{{\cal I}}{s} = s(\mathsf{R})$
  \item $\mwrst{\mathsf{c}}{{\cal I}}{s} = {\cal I}(\mathsf{c}) = \mathsf{c}$
  \item $\mwrst{\mathsf{p}}{{\cal I}}{s} = {\cal I}(\mathsf{p})$
  \item \label{item:threeval-apply}$\mwrst{(\mathsf{E}_1\ \mathsf{E}_2)}{{\cal I}}{s} = \bigwedge_{\preceq_{\pi}}\{\lsem \mathsf{E}_1 \rsem^{*}_s(\mathcal{I})(d) \mid d \in \lsem \rho\rsem, \lsem \mathsf{E}_2 \rsem^{*}_s({\cal I}) \preceq_{\rho} d\}$,
      for $\mathsf{E}_1\! :\! \rho \to \pi$ and $\mathsf{E}_2\! :\! \rho$
  \item $\mwrst{(\mathsf{E}_1\approx \mathsf{E}_2)}{{\cal I}}{s} = \begin{cases}
    \mathit{true},  & \text{if } \mwrst{\mathsf{E}_1}{{\cal I}}{s} = \mwrst{\mathsf{E}_2}{{\cal I}}{s} \\
    \mathit{false}, & \text{otherwise}
    \end{cases}$
  \item $\mwrst{(\sim \mathsf{E})}{{\cal I}}{s} =  (\mwrst{\mathsf{E}}{{\cal I}}{s})^{-1}$, with $\mathit{true}^{-1}\!=\!\mathit{false}$, $\mathit{false}^{-1}\!=\!\mathit{true}$
  and $\mathit{undef}^{-1}\!=\!\mathit{undef}$
  \item $\mwrst{(\mathsf{E}_1 \wedge \cdots \wedge \mathsf{E}_m)}{I}{s} =
    \bigwedge_{\leq_\bool}\{\mwrst{\mathsf{E}_1}{I}{s},\ldots,\mwrst{\mathsf{E}_m}{I}{s}\}$
\end{enumerate}
\end{definition}

Item~\ref{item:threeval-apply} is perhaps the most noteworthy.
To evaluate an expression $(\mathsf{E}_1\ \mathsf{E}_2)$, we cannot just take $\mwrst{\mathsf{E}_1}{{\cal I}}{s} $, which is a function $\mo{\rho}_{U_{\mathsf{P}}} \to \mos{\pi}_{U_{\mathsf{P}}}$, and apply it to $\mwrst{\mathsf{E}_2}{{\cal I}}{s}$, which is of type $\mos{\rho}_{U_{\mathsf{P}}}$.
Instead, we apply $\mwrst{\mathsf{E}_1}{{\cal I}}{s} $ to all ``two-valued extensions'' of  $\mwrst{\mathsf{E}_2}{{\cal I}}{s}$ and take the least precise element approximating all those results.
Our definition ensures that if $\mwrst{\mathsf{E}_2}{{\cal I}}{s}$ is a partial object, the result of the application
is the most precise outcome achievable by using information from all the two-valued extensions of the argument.

Application is always well-defined, \ie the set of two-valued extensions
of a three-valued element is always non-empty, as the following lemma suggests.
\begin{lemma}\label{two-valued-above-three-valued}
For every argument type $\rho$ and $d^* \in \mos{\rho}$, there exists $d \in \mo{\rho}$ such that $d^* \preceq_\rho d$.
\end{lemma}

Moreover, as the following lemma suggests, the above semantics (Definition~\ref{tuple-semantics})
is compatible with the standard semantics (see, Definition~\ref{standard-semantics})
when restricted to two-valued interpretations.

\begin{lemma}\label{mo-mos-coincide}
Let $\mathsf{P}$ be a program, $I \in H_\mathsf{P}$ and $s \in S_\mathsf{P}$.
Then, for every expression $\mathsf{E}$, $\mo{\mathsf{E}}_s(I) = \mos{\mathsf{E}}_s(I)$.
\end{lemma}

This three-valued valuation of bodies, immediately gives us a notion of three-valued model as well as a three-valued immediate consequence operator.

\begin{definition}\label{def:three-valued_model}
Let $\mathsf{P}$ be a program and ${\cal M}$ be a three-valued Herbrand interpretation of $\mathsf{P}$.
Then, ${\cal M}$ is a \emph{three-valued Herbrand model} of $\mathsf{P}$ iff for every rule
$\mathsf{p}\ \overline{\mathsf{R}} \lrule \mathsf{B}$ in $\mathsf{P}$
and for every Herbrand state $s$, $\mwrst{\mathsf{B}}{{\cal M}}{s} \leq_\bool \mwrst{\mathsf{p}\ \overline{\mathsf{R}}}{{\cal M}}{s}$.
\end{definition}

For the special case where in the above definition ${\cal M} \in H_\mathsf{P}$,
it is clear from Lemma~\ref{mo-mos-coincide} that
Definition~\ref{def:three-valued_model} coincides with
Definition~\ref{def:twovalued-models}.

\begin{definition}\label{def:three-valuedTP}
  Let $\mathsf{P}$ be a program. The \emph{three-valued immediate consequence operator}
  ${\cal T}_\mathsf{P} : {\cal H}_\mathsf{P} \to {\cal H}_\mathsf{P}$ is defined
  for every predicate constant $\mathsf{p} : \rho_1 \to \cdots \to \rho_n \to \bool$ in $\mathsf{P}$ and
  all $d_1 \in \mo{\rho_1},\ldots, d_n \in \mo{\rho_n}$, as:
  ${\cal T}_{\mathsf{P}}({\cal I})(\mathsf{p})\ \overline{d} =
        \bigvee_{\leq_\bool}\{
          \mwrst{\mathsf{B}}{{\cal I}}{s[\overline{\mathsf{R}}/\overline{d}]} \mid
                  \mbox{$s\in S_{\mathsf{P}}$ and
                        $(\mathsf{p}\ \overline{\mathsf{R}} \lrule \mathsf{B})$ in $\mathsf{P}$}\}$.
\end{definition}

The proof of the following proposition is similar to that of Proposition~\ref{model-iff-tp-prefixpoint}.
%\bart{All proofs are omitted, no?}
%
\begin{proposition}\label{three-valued-model-iff-tp-prefixpoint}
Let $\mathsf{P}$ be a program and ${\cal I} \in {\cal H}_\mathsf{P}$. Then, ${\cal I}$ is a three-valued
model of $\mathsf{P}$ if and only if ${\cal I}$ is a pre-fixpoint of ${\cal T}_\mathsf{P}$.
\end{proposition}

\section{Approximation Fixpoint Theory and the Stable Model Semantics}\label{sec:AFT_stable_model_semantics}
We now define the \emph{two-valued} and \emph{three-valued stable models} of a
program $\mathsf{P}$. To achieve this goal, we use the machinery of
\emph{approximation fixpoint theory
(AFT)}~\citep{DMT04Ultimateapproximationapplicationnonmonotonicknowledgerepresentation}.
In the rest of this section, we assume the reader has a basic familiarity
with~\citep{DMT04Ultimateapproximationapplicationnonmonotonicknowledgerepresentation}.
As mentioned before, the two-valued immediate consequence operator
$T_\mathsf{P}: H_{\mathsf{P}} \to H_{\mathsf{P}}$ can be non-monotone, meaning
it is not clear what its fixpoints of interest would be. The core idea behind
AFT is to ``approximate'' $T_{\mathsf{P}}$ with a function $\AP$ which is
$\preceq$-monotone. We can then study the fixpoints of $\AP$, which shed light
to the fixpoints of $T_{\mathsf{P}}$. While we already have such a candidate
function, namely ${\cal T}_\mathsf{P}$, AFT requires a function that works on
\emph{pairs} (of interpretations). Therefore, we show that there is a simple
isomorphism between three-valued relations and (appropriate) pairs of two-valued
ones. This isomorphism also exists between three-valued interpretations and
(appropriate) pairs of two-valued ones.
% This is captured by the following definitions and proposition.

%
\begin{definition}\label{def:consistent_lattice}\label{orderings_on_pairs}
Let $(L,\leq)$ be a complete lattice. We define $L^{c} =\{(x,y) \in L \times L \mid x \leq y\}$.
Moreover, we define the relations $\leq$ and  $\preceq$, so that
for all $(x,y),(x',y') \in L^c$:
%
% \begin{itemize}
$(x, y) \leq (x', y')$ iff $x \leq x'$ and $y \leq y'$, and
$(x, y) \preceq (x', y')$ iff $x \leq x'$ and $y' \leq y$.
% \end{itemize}
%
\end{definition}
\begin{proposition}\label{tau-isomorphism-preserves}
For every predicate type $\pi$ there exists a bijection $\tau_\pi: \mos{\pi} \to \cp{\mo{\pi}}$
with inverse $\tau^{-1}_\pi:  \cp{\mo{\pi}} \to \mos{\pi}$, that both preserve the orderings $\leq$ and $\preceq$
of elements between $\mos{\pi}$ and $\cp{\mo{\pi}}$. Moreover, there exists a bijection $\tau: {\cal H}_{\mathsf{P}} \to H^{c}_{\mathsf{P}}$ with inverse $\tau^{-1}: H^{c}_{\mathsf{P}} \to {\cal H}_{\mathsf{P}}$,
that both preserve the orderings $\leq$ and $\preceq$ between ${\cal H}_{\mathsf{P}}$ and
$H^{c}_{\mathsf{P}}$.
\end{proposition}

When viewing elements of $\mos{\pi}$ as partial sets and elements of $\mo{\pi}$
as sets, the isomorphism maps a partial set onto the pair with first
component all the \emph{certain} elements of the partial set (those mapped to
$\mtrue$) and second component all the \emph{possible} elements (those mapped
to $\mtrue$ or $\mundef$).
Using these bijections, we can now define  $\AP$ which, as we demonstrate,
is an approximator of $T_{\mathsf{P}}$. Intuitively, $\AP$ is the ``pair version of
${\cal T}_{\mathsf{P}}$'' (instead of handling three-valued interpretations, it handles pairs of two-valued ones).
\begin{definition}\label{def:ATP}
	For each program $\mathsf{P}$,
$\ATP: H^{c}_\mathsf{P} \to H^{c}_\mathsf{P}$ is defined as
$\ATP(I,J) = \tau({\cal T}_\mathsf{P}(\tau^{-1}(I,J)))$.
\end{definition}
\begin{lemma}\label{ATP_is_approximator_of_TP}
Let $\mathsf{P}$ be a program. In the terminology of \citet{DMT04Ultimateapproximationapplicationnonmonotonicknowledgerepresentation}, $\ATP: H^{c}_\mathsf{P} \to H^{c}_\mathsf{P}$
is a consistent
approximator of $T_\mathsf{P}$.
\end{lemma}

Since $\ATP$ is ``the pair version of ${\cal T}_{\mathsf{P}}$'', it is not a surprise that it also
captures all the three-valued models of $\mathsf{P}$.

\begin{lemma}\label{model-iff-atp-prefixpoint}
Let $\mathsf{P}$ be a program and $(I, J)\in H^{c}_\mathsf{P}$.
Then, $(I, J)$ is a pre-fixpoint of $\ATP$ if and only if $\tau^{-1}(I, J)$
is a three-valued model of $\mathsf{P}$.
\end{lemma}

Due to the above lemma, by stretching notation, when $(I,J)$ is a pre-fixpoint of $\ATP$
we will also say that $(I,J)$ is a model of $\mathsf{P}$.

The power of AFT comes from the fact that once an approximator is defined, it immediately defines a whole range of semantics.
In other words, there is no need to reinvent the wheel. The following definition summarizes the different induced semantics.
 \begin{definition}\label{def:AFTsemantics}
 	Let $\mathsf{P}$ be a program, $\ATP$ the induced approximator and $I, J \in H_\mathsf{P}$. We call:
 	\begin{itemize}
 		\item $(I,J)$ a \emph{three-valued supported model of $\mathsf{P}$} if it is a fixpoint of $\ATP$;
 		\item $(I,J)$ a \emph{three-valued stable model of $\mathsf{P}$} if it is a stable fixpoint of $\ATP$;
          that is, if $I=\lfp \ATP(\cdot,J)_1$ and $J=\lfp \ATP(I,\cdot)_2$, where $\ATP(\cdot,J)_1$ is the function that maps an interpretation $X$ to the first component of $\ATP(X,J)$, and similarly for $\ATP(I,\cdot)_2$;
 		\item $(I,J)$ the \emph{Kripke-Kleene model of $\mathsf{P}$} if it is the $\preceq$-least fixpoint of $\ATP$;
 		\item $(I,J)$ the \emph{well-founded model of $\mathsf{P}$} if it is the well-founded fixpoint of $\ATP$, \ie if it is the $\preceq$-least three-valued stable model.
 	\end{itemize}
 \end{definition}

Following the correspondence indicated by the isomorphism between pairs and three-valued interpretations,
we will also call $\mathcal{M}$ a \emph{three-valued stable model} of $\textsf{P}$ if $\tau(\mathcal{M})$ is a
\emph{three-valued stable model of $\mathsf{P}$}. If $\mathcal{M}$ is a three-valued stable model and $\mathcal{M} \in H_\mathsf{P}$ (i.e., $\mathcal{M}$ is actually two-valued), we will call $\mathcal{M}$ a \emph{stable model}
of $\mathsf{P}$.

\section{Properties of the Stable Model Semantics}\label{sec:properties}

In this section we discuss various properties of the stable model semantics of higher-order
logic programs, which demonstrate that the proposed approach is indeed an extension of
classical stable models.
In the following results we use the term ``classical stable models'' to
refer to stable models in the sense of~\citep{GL88StableModelSemanticsLogicProgramming},
``classical three-valued stable models'' to refer to stable models in the sense of~\citep{P90Well-FoundedSemanticsCoincidesThree-ValuedStableSemantics}
and the term ``(three-valued) stable models'' to refer to the present semantics.

\begin{theorem}\label{coincides-with-classical-stable-models}
Let $\mathsf{P}$ be a propositional logic program. Then, ${\cal M}$ is a (three-valued)
stable model of $\mathsf{P}$ iff ${\cal M}$ is a classical (three-valued) stable model of $\mathsf{P}$.
\end{theorem}

A crucial property of classical (three-valued) stable models is that they are \emph{minimal} Herbrand models~\citep[Theorem 1]{GL88StableModelSemanticsLogicProgramming}
and~\citep[Proposition~3.1]{P90Well-FoundedSemanticsCoincidesThree-ValuedStableSemantics}.
This property is preserved by our extension.
\begin{theorem}\label{stable-models-are-minimal}
All (three-valued) stable models of a $\HOL$ program $\mathsf{P}$ are
$\leq$-minimal models of~$\mathsf{P}$.
\end{theorem}

It is a well-known result in classical logic programming that if the well-founded model of a first-order program
is two-valued, then that model is its unique classical stable model~\cite[Corollary~5.6]{VRS88UnfoundedSetsWell-FoundedSemanticsGeneralLogic}.
This property generalizes in our setting.
\begin{theorem}\label{exact-wf-unique-stable}
Let $\mathsf{P}$ be a ${\cal HOL}$ program. If the well-founded model of
$\mathsf{P}$ is two-valued, then this is also its unique stable model.
\end{theorem}

A broadly studied subclass of first-order logic programs with negation, is that of \emph{stratified
logic programs}~\citep{ABW88TowardsTheoryDeclarativeKnowledge}. It is a well-known result that
if a logic program is stratified, then it has a two-valued well-founded model which is also its
unique classical stable model \citep[Corollary~2]{GL88StableModelSemanticsLogicProgramming}.
We extend the class of stratified programs to the higher-order case and generalize the
aforementioned result.
\begin{definition}
A ${\cal HOL}$ program $\mathsf{P}$ is called \emph{stratified} if there is a function $S$ mapping
predicate constants to natural numbers, such that for each rule
$\mathsf{p} \ \overline{\mathsf{R}} \leftarrow \mathsf{L}_1 \wedge \cdots \wedge \mathsf{L}_m$
and any $i\in\{1,\ldots, m\}$:
\begin{itemize}
  \item $S(\mathsf{q})\leq S(\mathsf{p})$ for every predicate constant $\mathsf{q}$ occurring in $\mathsf{L}_i$.
  \item If $\mathsf{L}_i$ is of the form $\sim \!\mathsf{E}$, then $S(\mathsf{q})< S(\mathsf{p})$ for each predicate constant $\mathsf{q}$ occurring~in~$\mathsf{E}$.
  \item For any subexpression of $\mathsf{L}_i$ of the form $(\mathsf{E}_1~\mathsf{E}_2)$, $S(\mathsf{q})< S(\mathsf{p})$ for every predicate constant $\mathsf{q}$ occurring in $\mathsf{E}_2$.
\end{itemize}
\end{definition}
For readers familiar with the standard definitions of stratification in first-order logic
programs, the last item might be somewhat surprising. What it says is that the
stratification function should not only increase because of negation, but also
because of higher-order predicate application. The
intuitive reason for this is that (as also noted in the introduction of the present paper)
one can define a higher-order predicate
which is identical to negation, for example, by writing \lstinline|neg P :- ~P|.
As a consequence, it is reasonable
to assume that predicates occurring inside an application of \lstinline|neg|
should be treated similarly to predicates appearing inside the negation symbol.

\begin{theorem}\label{well-founded-exact-on-stratified}
Let $\mathsf{P}$ be a stratified $\HOL$ program. Then, the well-founded model of $\mathsf{P}$ is two-valued.
\end{theorem}

By the above theorem and Theorem~\ref{exact-wf-unique-stable}, every stratified $\HOL$ program has
a unique two-valued stable model.

\section{Additional Examples}\label{sec:applications}
In this section we present two examples of how higher-order logic programming can be used.
First, we showcase reasoning problems arising from the field of abstract argumentation, next we model
a $\mathsf{PSPACE}$-complete problem known as Generalized Geography.

\subsection{Abstract Argumentation}
In what follows, we present a set of standard definitions from the field of abstract argumentation \mycitep{AF}.
Listing~\ref{list:AF} contains direct translations of these definitions into our framework; the line numbers with each definition refer to Listing \ref{list:AF}.
Listing~\ref{list:AFex} illustrate how these definitions can be used to solve reasoning problems with argumentation.

An \emph{abstract argumentation framework}~(AF) $\Theta$
is a directed graph $(A,E)$ in which the nodes $A$ represent arguments and the edges in $E$ represent attacks between arguments.
We say that $a$ \emph{attacks} $b$ if $(a,b)\in E$.
A set  $S\subseteq A$ \emph{attacks} $a$ if some $s\in S$ attacks $a$ (Line~\ref{line:attacks}).
A set $S\subseteq A$ \emph{defends} $a$ if  it attacks all attackers of $a$ (Line~\ref{line:defends}).
An \emph{interpretation} of an AF $\Theta=(A,E)$ is a subset $S$ of $A$.
There exist many different semantics of AFs that each define different sets of acceptable arguments according to different standards or intuitions.
The major semantics for argumentation frameworks can be formulated using two operators:
the \emph{characteristic function} $F_\Theta$ (Line~\ref{line:f}) mapping an interpretation $S$ to
\[F_\Theta(S) = \{a\in A\mid S \text{ defends } a\}\]
and the operator $U_\Theta$ ($U$ stands for unattacked; Line~\ref{line:unattacked}) that maps an interpretation $S$ to
\[U_\Theta(S) = \{a\in A\mid a\text{ is not attacked by }S\}.\]
The \emph{grounded extension} of $\Theta$ is defined inductively as the set of all arguments defended by the grounded extension (Line~\ref{line:grounded}), or alternatively, as the least fixpoint of $F_\Theta$, which is a monotone operator.
The operator $U_\Theta$ is an anti-monotone operator; its fixpoints are called \emph{stable extensions} of $\Theta$ (Line~\ref{line:stable}).
An interpretation $S$ is \emph{conflict-free} if it is a postfixpoint of $U_\Theta$ (\ie if $S\subseteq U_\Theta(S)$; Line~\ref{line:conflfree}).
A \emph{complete extension} is a conflict-free fixpoint of $F_\Theta$ (Line~\ref{line:complete}).
An interpretation is \emph{admissible} if it is a conflict-free postfixpoint of $F_\Theta$ (Line~\ref{line:admissable}).  A \emph{preferred extension} is a $\subseteq$-maximal complete extension (Line~\ref{line:preferred}).

Listing~\ref{list:AFex} shows how these definitions can be used for reasoning problems related to argumentation.
There, we search for an argumentation framework with five elements where the grounded extension does not equal the intersection of all stable extensions.

\begin{lstlisting}[label=list:AF,caption=Second-order definitions of abstract argumentation concepts.]
% A is a set of arguments; E subset A x A is the attack relation
attacks A E S X :- (subset S A), (S Y), (E Y X)|\label{line:attacks}|
nondefends A E S X :- (subset S A), (A Y), (E Y X), ~(attacks A E S Y)
defends A E S X :- (subset S A), (A X), ~(nondefends A E S X)|\label{line:defends}|
f A E S X :- (defends A E S X)|\label{line:f}|
u A E S X :- (subset S A), (A X), ~(attacks A E S X)|\label{line:unattacked}|
% grounded A E X means: X is an element of the grounded extension
grounded A E X  :- f A E (grounded A E) X |\label{line:grounded}|
stable A E S :- (equal S (u A E S))|\label{line:stable}|
conflFree A E S :- (subset S (u A E S))|\label{line:conflfree}|
complete A E S :- (conflFree A E S), (equal S (f A E S))|\label{line:complete}|
admissable A E S  :- (conflFree A E S), (subset S (f A E S))|\label{line:admissable}|
preferred A E S :-  maximal subset (complete A E) S |\label{line:preferred}|
\end{lstlisting}
\stickListings
\begin{lstlisting}[label=list:AFex,caption=Toy reasoning problem for abstract argumentation.]
arg a. arg b. arg c. arg d. arg e.
attacks X Y  :- arg X, arg Y, ~(nattacks X Y).
nattacks X Y :- arg X, arg Y, ~(attacks X Y).

ncautiousStable X :- arg X, stable arg attacks S, ~(S X).
cautiousStable X :- arg X, ~(ncautiousStable X).
p :- ~p, equal cautiousStable (grounded arg attacks).
\end{lstlisting}

\subsection{(Generalized) Geography}
Generalized geography is a two-player game that is played on a graph. Two
players take turn to form a simple path (\ie a path without cycles) through the
graph. The first player who can no longer extend the currently formed simple
path loses the game. The question whether a given node in a given graph is a
winning position in this game (\ie whether there is a winning strategy) is
well-known to be $\mathsf{PSPACE}$-hard (see, \eg the proof of
\citet{LS80GOPolynomial-SpaceHard}). This game can be modelled in our language
very compactly: Line~\ref{line:winning} in Listing~\ref{listing:GG} states that
\lstinline|X| is a winning node in the game \lstinline|V, E| if there is an
outgoing edge from \lstinline|X| that leads to a non-winning position in the
induced graph obtained by removing \lstinline|X| from \lstinline|V|. This
definition makes use of the notion of an induced subgraph, which has a very
natural higher-order definition, which in turn makes use of various other
generic predicates about sets (see Listings \ref{list:generic} and
\ref{list:generic2}).

\begin{lstlisting}[caption=Winning positions in the Generalized Geography game.,label=listing:GG]
% X is a winning position in the GG game (V,E)
winning V E X :- (E X Y), ~(X=Y), equal (remove V X) V', |\\\newlineInListing|inducedGraph V E V' E', |\\\newlineInListing|~(winning V' E' Y).  |\label{line:winning}|
% (V',E') is the induced graph by restricting (V,E) to V'
inducedGraph V E V' E' :- subset V' V, |\\\newlineInListing| equal E' (intersection E (square V'))
\end{lstlisting}
\stickListings
\begin{lstlisting}[caption=More generic definitions.,label=list:generic2]
% X is in the union of P and Q
union P Q X :- P X.
union P Q X :- Q X.
% X is in the intersection of P and Q
intersection P Q X :- P X, Q X.
% Y is in the set obtained from P by removing X  (P \ {X})
remove P X Y :- P Y, ~(X=Y).
% (X,Y) is in the square of P (cartesian product of P with itself)
square P X Y :- P X, P Y.
\end{lstlisting}

\section{Related and Future Work}
There are many extensions of standard logic programming under the stable model
semantics that are closely related to our current work. One of them is the
extension of logic programming with \emph{aggregates}, which most solvers
nowadays support. Aggregates are special cases of second-order functions and
have been studied using AFT
\citep{PDB07Well-foundedstablesemanticslogicprogramsaggregates,VBD22AnalyzingSemanticsAggregateAnswerSetProgramming}
and in fact our semantics of application can be viewed as a generalization of
the \emph{ultimate approximating aggregates} of
\citet{PDB07Well-foundedstablesemanticslogicprogramsaggregates}. Also,
higher-order logic programs have been studied through this fixpoint theoretic
lens. \citet{DvJD15Semanticstemplatescompositionalframeworkbuildinglogics}
defined a logic for templates, which are second-order definitions, for which
they use a well-founded semantics. This idea was generalized to arbitrary
higher-order definitions in the next year
\citep{DvBJD16CompositionalTypedHigher-OrderLogicDefinitions}. While they apply
AFT in  the same space of three-valued higher order functions as  we do, a
notable difference is that they use the so-called \emph{ultimate approximator},
resulting in a semantics that does not coincide with the standard semantics for
propositional programs whereas our semantics does (see Theorem
\ref{coincides-with-classical-stable-models}).

In 2018, a well-founded semantics for higher-order logic programs was developed using AFT \citep{CRS18ApproximationFixpointTheoryWell-FoundedSemanticsHigher-Order}.
There are two main ways in which that work differs from ours.
The first, and arguably most important one, is how the three-valued semantics of types is defined.
While in our framework $\mos{\rho\to\bool}$ consists of all functions from $\mo{\rho}$ to $\mos{\bool}$, in their framework $\mos{\rho\to\bool}$ would consist of all $\preceq$-monotonic functions from $\mos{\rho}$ to $\mos{\bool}$.
This results, in their case, to more refined, but less precise, and more complicated, approximations.
As a result, an extension of AFT needed to be developed to accommodate this.
In our current work, we show that we can stay within standard AFT, but to achieve this, we needed to develop a new three-valued semantics
of function application; see Item~\ref{item:threeval-apply} in Definition~\ref{tuple-semantics}.
The formal relationship between the two approaches remains to be further investigated.
The second way in which our work differs from that of \citet{CRS18ApproximationFixpointTheoryWell-FoundedSemanticsHigher-Order} is the treatment of existential predicate variables of type $\pi$ that appear in the bodies of rules; we consider such variables to range over $\mo{\pi}$, while \citeauthor{CRS18ApproximationFixpointTheoryWell-FoundedSemanticsHigher-Order} allow them to range over $\mos{\pi}$.
A consequence of this choice is that arguably, the well-founded semantics of \citep{CRS18ApproximationFixpointTheoryWell-FoundedSemanticsHigher-Order}
does not always behave as expected; even for simple non-recursive programs such as \lstinline|p :- R, ~R|, the meaning of the defined predicates
 is not guaranteed to be two-valued.
This is not an issue with their extension of AFT, but rather with the precise way their approximator is defined.
While we believe it would be possible to solve this issue by changing the definition of the approximator in
\citet{CRS18ApproximationFixpointTheoryWell-FoundedSemanticsHigher-Order}, this
issue needs to be further investigated.
Importantly, we showed in Theorem \ref{well-founded-exact-on-stratified} that issues such as the one just mentioned, cannot arise in our new semantics.

Recently, there have also been some extensions of logic programming that allow
second-order quantification \emph{over} answer sets. This idea was first
referred to as \emph{stable-unstable semantics}
\citep{BJT16Stable-unstablesemanticsBeyondNPnormallogic} and later also as
\emph{quantified ASP} \citep{ART19BeyondNPQuantifyingoverAnswerSets}. A related
formalism is \emph{ASP with quantifiers}
\citep{FLRSS21PlanningIncompleteInformationQuantifiedAnswerSet}, which can be
thought of as a \emph{prenex} version of quantified ASP, consisting of a single
logic program preceded by a list of quantifications. A major advantage of those
lines of work is that they come with efficient implementations
\citep{J22ImplementingStable-UnstableSemanticsASPTOOLSClingo,FMR23EfficientSolverASPQ}
and applications
\citep{ACRT22SolvingProblemsPolynomialHierarchyASPQ,BMR22ModellingOutlierDetectionProblemASPQ,FLRSS21PlanningIncompleteInformationQuantifiedAnswerSet}.
An advantage of true higher-order logic programming (which allows for defining
higher-order predicates) is the  potential for abstraction and reusability
(following the lines of thought of the ``templates'' work referred to above). As
an example, consider our max-clique application from Section
\ref{sec:motivating}. While it is perfectly possible to express this in
stable-unstable semantics or quantified ASP, such encodings would have
\emph{two} definitions of what it means to be a clique: one for the actual
clique to be found and one inside the oracle call that checks for the
non-existence of a larger clique. In our approach, the definition of clique is
given only once and used for these two purposes by giving it as an argument to
the higher-order \lstinline|maximal| predicate. Moreover, the definition of the
\lstinline|maximal| predicate can be reused in future applications where maximal
(with respect to some given order) elements of some set are sought.

There are several future directions that we feel are worth pursuing. In
particular, it would be interesting to investigate efficient implementation
techniques for the proposed stable model semantics. As the examples of the paper
suggest, even an implementation of second-order stable models, would give a
powerful and expressive system. Another interesting research topic is the
characterization of the expressive power of higher-order stable models. As
proven in~\citep{CNR19ExpressivePowerHigher-OrderDatalog}, positive $k$-order
Datalog programs over ordered databases, capture $(k-1)$-$\mathsf{EXPTIME}$ (for
all $k\geq 2$). We believe that the addition of stable negation will result in
greater expressiveness (for example, the ordering restriction on the database
could be lifted), but this needs to be further investigated.

% bb_refs is a read-only reference file. I also included a shortened version in which all refs are shorteneed if this would be needed for page
% otherrefs is for refs not included in any of the other bib files
\bibliography{bb_includes/bb_refs_shortened,otherrefs}

\ifincludeappendix
\appendix
\section{Proofs of Section~\ref{sec:two-valued-semantics} and Section~\ref{sec:three-valued-semantics}}

In the following, we provide the proofs of a few propositions contained in Sections~4 and~5. Notice that most of the results of these two sections are rather straightforward. Propositions~\ref{semantics_of_types_lattice} and \ref{semantics_of_types_lattice_cpo} are algebraic consequences of Definitions \ref{def:orders} and \ref{def:orders_three-valued}, respectively.

\begin{reproposition}{semantics_of_types_lattice}
For every predicate type $\pi$, $(\mo{\pi}, \leq_\pi)$ is a complete lattice.
\end{reproposition}
\begin{proof}
	We proceed by induction on the predicate type $\pi$.

	Let $\pi= o$. Clearly, the set $\mo{o} =\{\mathit{true}, \mathit{false} \}$ with the order $\leq_o$ is a complete lattice, with bottom element $\mathit{false}$ and top element $\mathit{true}$.

	Now let $\pi= \rho \to \pi'$, and assume $\mo{\pi'}$ is a complete lattice. We have to show that the set of functions $\mo{\rho} \to	\mo{\pi'}$ with the order $\leq_{\rho\to\pi}$ is a complete lattice. Since $\mo{\pi'}$ is a complete lattice, for each subset $S\subseteq \mo{\pi}$ we can define $f_{\bigwedge S}, f_{\bigvee S} \in \mo{\pi}$ by $f_{\bigwedge S}(x):=\bigwedge\{g(x)\mid g\in S\}$ and  $f_{\bigvee S}(x):=\bigvee\{g(x)\mid g\in S\}$, respectively. By the definition of $\leq_{\rho\to\pi'}$, it is immediate to see that $\bigwedge S=f_{\bigwedge S}$ and $\bigvee S=f_{\bigvee S}$. Hence, $\mo{\pi}$ is a complete lattice, as desired.
\end{proof}

The following proposition draws the correspondence between the two-valued Herbrand models of a program and the pre-fixpoints of its immediate consequence operator.

\begin{reproposition}{model-iff-tp-prefixpoint}
Let $\mathsf{P}$ be a program and $I \in H_\mathsf{P}$. Then, $I$ is a model of $\mathsf{P}$
iff $I$ is a pre-fixpoint of $\TP$.
\end{reproposition}
\begin{proof}
Suppose first that $I$ is a model of $\mathsf{P}$. Suppose, for the sake of contradiction, that $I$ is not a pre-fixpoint of $\TP$.
That means that there exists a predicate constant
$\mathsf{p} : \rho_1 \to \cdots \to \rho_n \to \bool$ and a $\overline{d} \in \mo{\rho_1} \times \cdots \times \mo{\rho_n}$
such that $\TP(I)(\mathsf{p})\ \overline{d}\not\leq I(\mathsf{p})\ \overline{d}$. Then, $\TP(I)(\mathsf{p})\ \overline{d}$ must be $\mathit{true}$, so by the definition
of $\TP$ there is a Herbrand state $s$ and a rule $\mathsf{p}\ \overline{\mathsf{R}} \lrule \mathsf{B}$ such that
$\mwrs{\mathsf{B}}{I}{s[\overline{\mathsf{R}}/\overline{d}]} = \mathit{true}$. Since $I$ is a model,
$\mwrs{\mathsf{p}\ \overline{\mathsf{R}}}{I}{s[\overline{\mathsf{R}}/\overline{d}]} = \mathit{true}$.
Then, $I(\mathsf{p})\ \overline{d}=\mathit{true}$, which is a contradiction.

In the other direction, suppose that $I$ is a pre-fixpoint of $\TP$.
Suppose $s$ is a Herbrand state and $\mathsf{p}\ \mathsf{R}_1\ldots\mathsf{R}_n \lrule \mathsf{B}$ is a rule such that $\mwrs{\mathsf{B}}{I}{s} = \mathit{true}$.
Then, $\TP(I)(\mathsf{p})\ s(\mathsf{R}_1)\ldots s(\mathsf{R}_n) = \mathit{true}$.
Since $I$ is a pre-fixpoint of $\TP$, $I(\mathsf{p})\ s(\mathsf{R}_1)\ldots s(\mathsf{R}_n) = \mathit{true}$,
which implies that $\mwrs{\mathsf{p}\ \mathsf{R}_1\ldots\mathsf{R}_n}{I}{s} = \mathit{true}$.
\end{proof}

\begin{reproposition}{semantics_of_types_lattice_cpo}
For every predicate type $\pi$, $(\mos{\pi}, \leq_\pi)$
is a complete lattice and $(\mos{\pi}, \preceq_\pi)$ is a complete meet-semilattice (\ie
every non-empty subset of $\mos{\pi}$ has a $\preceq_\pi$-greatest lower bound).
\end{reproposition}
\begin{proof}
We proceed by induction on the predicate type $\pi$.

If $\pi=o$, then obviously $(\mos{\pi}, \leq_\pi)$
is a complete lattice and $(\mos{\pi}, \preceq_\pi)$ is a complete meet-semilattice.

 Let now $\pi= \rho\to \pi'$, and assume that $(\mos{\pi'}, \leq_\pi)$
is a complete lattice and $(\mos{\pi'}, \preceq_\pi)$ is a complete meet-semilattice. We show that $(\mos{\pi}, \leq_\pi)$
is a complete lattice. The proof that $(\mos{\pi'}, \preceq_\pi)$ is a complete meet-semilattice is analogous and omitted. Let $S$ be a subset of $ \mo{\rho}_{U_{\mathsf{P}}} \to \mos{\pi'}_{U_{\mathsf{P}}}$. By induction hypothesis, we can define $f_{\bigwedge S}, f_{\bigvee S} \in \mo{\pi}$ by $f_{\bigwedge S}(x):=\bigwedge\{g(x)\mid g\in S\}$ and  $f_{\bigvee S}(x):=\bigvee\{g(x)\mid g\in S\}$, respectively. By the definition of $\leq_{\rho\to\pi'}$, it is immediate to see that $\bigwedge S=f_{\bigwedge S}$ and $\bigvee S=f_{\bigvee S}$. Hence, $\mo{\pi}$ is a complete lattice, as desired.
\end{proof}

\begin{relemma}{two-valued-above-three-valued}
For every argument type $\rho$ and $d^* \in \mos{\rho}$, there exists $d \in \mo{\rho}$ such that $d^* \preceq_\rho d$.
\end{relemma}
\begin{proof}
For the case $\rho=\iota$, we have $\mo{\iota} = \mos{\iota}$, so that $d^* \in \mo{\iota}$ and $d^* \preceq_\iota d^*$.
On the other hand, if $\rho = \pi$, $\rho$ would be of the form $\rho_1\to\ldots\to\rho_n\to o$.
We define $d \in \mo{\rho}$ such that for all $x_1 \in \mo{\rho_1},\ldots,x_n \in \mo{\rho_n}$:
\[
	d(\overline{x})=\begin{cases}
		d^*(\overline{x}), & \text{if } d^*(\overline{x})\in \{\mathit{false}, \mathit{true}\} \\
		\mathit{false}, & \text{otherwise}
	\end{cases}
\]
It is easy to see that $d \preceq_\rho d^*$.
\end{proof}

In order to establish Lemma~\ref{mo-mos-coincide}, we need to first show an auxiliary one:
\begin{lemma}\label{preceq-over-two-valued}
Let $\rho$ be an argument type. Then, $(\preceq_\rho)$ restricted to $\mo{\rho}$ is the trivial ordering, \ie
for all $d, d'\in \mo{\rho}$, $d \preceq_\rho d'$ if and only if $d=d'$.
\end{lemma}
\begin{proof}
By induction on the argument type $\rho$.
If $\rho = \iota$ the result follows from the definition of $\preceq_\iota$ over $\mos{\iota}$.
If $\rho = o$, it can be established by case analysis.
Suppose now that $\rho = \rho'\to\pi$ and the statement holds for $\pi$.
Let $d, d' \in \mo{\rho'\to\pi}$ such that $d \preceq_{\rho'\to\pi} d'$.
By the definition of $\preceq_{\rho'\to\pi}$, we have $d(x) \preceq_{\pi} d'(x)$ for all $x \in \mo{\rho'}$.
By the induction hypothesis, we have $d(x) = d'(x)$ for all $x \in \mo{\rho'}$.
Therefore, $d = d'$.
\end{proof}

\begin{relemma}{mo-mos-coincide}
Let $\mathsf{P}$ be a program, $I \in H_\mathsf{P}$ and $s \in S_\mathsf{P}$.
Then, for every expression $\mathsf{E}$, $\mo{\mathsf{E}}_s(I) = \mos{\mathsf{E}}_s(I)$.
\end{relemma}
\begin{proof}
By induction on the structure of $\mathsf{E}$. The only non-trivial case is when $\mathsf{E}$ is of the form $(\mathsf{E}_1~\mathsf{E}_2)$.
By Lemma~\ref{preceq-over-two-valued},
$\{\lsem \mathsf{E}_1 \rsem^{*}_s(I)(d) \mid d \in \lsem \rho\rsem, \lsem \mathsf{E}_2 \rsem^{*}_s({\cal I}) \preceq_{\rho} d\} =
\{\lsem \mathsf{E}_1 \rsem^{*}_s(I)(\lsem \mathsf{E}_2 \rsem^{*}_s({\cal I}))\}$.
Therefore,
\[
  \begin{array}{rlll}
    [\mos{(\mathsf{E}_1~\mathsf{E}_2)}_s(I) & = & \lsem \mathsf{E}_1 \rsem^{*}_s(I)(\lsem \mathsf{E}_2 \rsem^{*}_s({\cal I})) \\
    & = & \lsem \mathsf{E}_1 \rsem_s(I)(\lsem \mathsf{E}_2 \rsem_s({\cal I})) & \mbox{(Induction Hypothesis)} \\
    & = & \mo{(\mathsf{E}_1~\mathsf{E}_2)}_s(I)
  \end{array}
\]
This completes the proof of the lemma.
\end{proof}

% \begin{lemma}\label{ATP_is_Fitting_monotonic}
% %
% Let $\mathsf{P}$ be a program and let  $(I_1, J_1), (I_2, J_2) \in H_\mathsf{P} \otimes  H_\mathsf{P}$.
% If $(I_1, J_1) \preceq (I_2, J_2)$ then $A_{T_\mathsf{P}}(I_1, J_1) \preceq A_{T_\mathsf{P}}(I_2,J_2)$.
% %
% \end{lemma}

% \begin{definition}
% 	%
% 	For every predicate type $\pi$, we define the functions
% 	$\tau_\pi: \mos{\pi} \to \cp{\mo{\pi}}$ and
% 	$\tau^{-1}_\pi: \cp{\mo{\pi}} \to \mos{\pi}$,
% 	as follows.
% 	%
% 	\begin{itemize}
% 	  \item $\tau_o(\mfalse) = (\mfalse, \mfalse)$, $\tau_o(\mtrue) = (\mtrue, \mtrue)$,
% 							$\tau_o(\mundef) = (\mfalse, \mtrue)$
% 	  \item $\tau_{\iota}(d) = (d, d)$
% 	  % \item $\tau_{\rh_1. \o_1 \rightarrow \cdots \rightarrow \rho_n \rightarrow o}(f) =
% 	  %       (\lambda d_1. \cdots \lambda d_n . [\tau_o(f \ d_1 \cdots d_n)]_1,
% 	  %        \lambda dcdots \lambda d_n . [\tau_o(f \, d_1 \cdots d_n)]_2)$
% 	  \item $\tau_{\rho \to \pi}(f) = (\lambda d. [\tau_\pi(f(d))]_1, \lambda d. [\tau_\pi(f(d))]_2)$

% 	\end{itemize}
% 	%
% 	and
% 	%
% 	\begin{itemize}
% 	%
% 	\item $\tau^{-1}_o(\mfalse,\mfalse) = \mfalse$, $\tau^{-1}_o(\mtrue,\mtrue) = \mtrue$,
% 								  $\tau^{-1}_o(\mfalse, \mtrue) = \mundef$
% 	\item $\tau^{-1}_{\iota}(d,d) = d$
% 	% \item $\tau^{-1}_{\rho_1 \rightarrow \cdots \rightarrow \rho_n \rightarrow o}(f_1, f_2) = \lambda d_1. \cdots
% 	%       \lambda d_n.  \tau_o^{-1}(f_1 \, d_1 \cdots d_n, f_2 \, d_1 \cdots d_n)$.
% 	\item $\tau^{-1}_{\rho \to \pi}(f_1, f_2) = \lambda d. \tau^{-1}_{\pi}(f_1(d), f_2(d))$
% 	%
% 	\end{itemize}
% 	%
% \end{definition}

\section{Proofs of Section~\ref{sec:AFT_stable_model_semantics}}
In this appendix, we collect the proofs of the results of Section~\ref{sec:AFT_stable_model_semantics}.
We start with Proposition~\ref{tau-isomorphism-preserves}, which concerns the existence, for every predicate type $\pi$, of an  isomorphism,
\ie an order-preserving bijection, between the set $\mos{\pi}$ of three-valued meanings and the set $\cp{\mo{\pi}}$ of  pairs of two-valued ones.
We first provide the definition of such functions, and then we prove they are indeed isomorphisms in Proposition~\ref{tau-isomorphism-preserves}.
\begin{definition}\label{tau-definition}
For every predicate type $\pi$, we define the functions
$\tau_\pi: \mos{\pi} \to \cp{\mo{\pi}}$ and
$\tau^{-1}_\pi: \cp{\mo{\pi}} \to \mos{\pi}$,
as follows:
\begin{itemize}
  \item $\tau_o(\mfalse) = (\mfalse, \mfalse)$, $\tau_o(\mtrue) = (\mtrue, \mtrue)$,
                        $\tau_o(\mundef) = (\mfalse, \mtrue)$
  % \item $\tau_{\iota}(d) = (d, d)$
  % \item $\tau_{\rh_1. \o_1 \rightarrow \cdots \rightarrow \rho_n \rightarrow o}(f) =
  %       (\lambda d_1. \cdots \lambda d_n . [\tau_o(f \ d_1 \cdots d_n)]_1,
  %        \lambda dcdots \lambda d_n . [\tau_o(f \, d_1 \cdots d_n)]_2)$
  \item $\tau_{\rho \to \pi}(f) = (\lambda d. [\tau_\pi(f(d))]_1, \lambda d. [\tau_\pi(f(d))]_2)$

\end{itemize}
and
\begin{itemize}
\item $\tau^{-1}_o(\mfalse,\mfalse) = \mfalse$, $\tau^{-1}_o(\mtrue,\mtrue) = \mtrue$,
                              $\tau^{-1}_o(\mfalse, \mtrue) = \mundef$
% \item $\tau^{-1}_{\iota}(d,d) = d$
% \item $\tau^{-1}_{\rho_1 \rightarrow \cdots \rightarrow \rho_n \rightarrow o}(f_1, f_2) = \lambda d_1. \cdots
%       \lambda d_n.  \tau_o^{-1}(f_1 \, d_1 \cdots d_n, f_2 \, d_1 \cdots d_n)$.
\item $\tau^{-1}_{\rho \to \pi}(f_1, f_2) = \lambda d. \tau^{-1}_{\pi}(f_1(d), f_2(d))$
\end{itemize}
\end{definition}

The functions $\tau_\pi$ defined above, can easily be extended to a function between ${\cal H}_{\mathsf{P}}$ and $H^{c}_\mathsf{P}$: given ${\cal I} \in {\cal H}_{\mathsf{P}}$, we define $\tau({\cal I}) = (I,J)$, where
for every predicate constant $\mathsf{p}:\pi$ it holds $I(\mathsf{p}) = [\tau_{\pi}({\cal I}(\mathsf{p}))]_1$ and $J(\mathsf{p}) = [\tau_{\pi}({\cal I}(\mathsf{p}))]_2$.
Conversely, given a pair $(I,J) \in \cp{H_\mathsf{P}}$,
we define the three-valued Herbrand interpretation $\tau^{-1}(I,J)$, for every predicate constant $\mathsf{p}:\pi$, as follows:
$\tau^{-1}(I,J)(\mathsf{p}) = \tau^{-1}_{\pi}(I(\mathsf{p}),J(\mathsf{p}))$.

\begin{reproposition}{tau-isomorphism-preserves}
  For every predicate type $\pi$ there exists a bijection $\tau_\pi: \mos{\pi} \to \cp{\mo{\pi}}$
  with inverse $\tau^{-1}_\pi:  \cp{\mo{\pi}} \to \mos{\pi}$, that both preserve the orderings $\leq$ and $\preceq$
  of elements between $\mos{\pi}$ and $\cp{\mo{\pi}}$. Moreover, there exists a bijection $\tau: {\cal H}_{\mathsf{P}} \to H^{c}_{\mathsf{P}}$
  with inverse $\tau^{-1}: H^{c}_{\mathsf{P}} \to {\cal H}_{\mathsf{P}}$,
  that both preserve the orderings $\leq$ and $\preceq$ between ${\cal H}_{\mathsf{P}}$ and
  $H^{c}_{\mathsf{P}}$.
\end{reproposition}

\begin{proof}
	Consider the functions in Definition~\ref{tau-definition}.
	Let $\pi$ be a predicate type. It follows easily from the definition that $\tau_\pi,\tau^{-1}_\pi$ are well defined functions and a formal proof using
	induction on the type structure is omitted.

	We show that they are also order-preserving, $\tau_\pi$ is a bijection and $\tau^{-1}_\pi$ is the inverse.
	Specifically, for every $f,g \in \mos{\pi}$
	and for every $(f_1,f_2),(g_1,g_2) \in \cp{\mo{\pi}}$, we show that
	the following statements hold:
	\begin{enumerate}
	%   \item $\tau_{\pi}(f) \in  (\lsem \pi \rsem_D^{\mathsf{ma}} \otimes \lsem \pi \rsem_D^{\mathsf{am}})$
	%         and $\tau^{-1}_{\pi}(f_1,f_2) \in \lsem \pi \rsem_D$.
	\item If $f \preceq_{\pi} g$ then $\tau_\pi(f) \preceq_{\pi} \tau_\pi(g)$.
	\item If $f \leq_{\pi} g$ then $\tau_\pi(f) \leq_{\pi} \tau_\pi(g)$.
	\item If $(f_1,f_2) \preceq_{\pi} (g_1,g_2)$ then $\tau_\pi^{-1}(f_1,f_2) \preceq_{\pi} \tau_\pi^{-1}(g_1,g_2)$.
	\item If $(f_1,f_2) \leq_{\pi} (g_1,g_2)$ then $\tau^{-1}_\pi(f_1,f_2) \leq_{\pi} \tau^{-1}_\pi(g_1,g_2)$.
	\item $\tau^{-1}_{\pi}(\tau_{\pi}(f))=f$
	\item $\tau_{\pi}(\tau^{-1}_{\pi}(f_1,f_2))=(f_1,f_2)$
	\end{enumerate}
	We will use structural induction on the types. For the base types $o$ and $\iota$ the proof is trivial
	for all the statements. Consider the general case of $\pi: \rho_1 \to \pi_2$.

\noindent	{\em Statement 1:} Assume that $f \preceq g$. Then, for any $d \in \mo{\rho_1}$ it is $f(d) \preceq_{\pi_2} g(d)$
	and by induction hypothesis $\tau_{\pi_2}(f(d)) \preceq_{\pi_2}\tau_{\pi_2}(g(d))$. Therefore, $[\tau_{\pi_2}(f(d))]_1 \leq_{\pi_2} [\tau_{\pi_2}(g(d))]_1$ and
	by abstracting $\lambda d. [\tau_{\pi_2}(f(d))]_1 \leq_{\pi_2} \lambda d. [\tau_{\pi_2}(g(d))]_1$. Similarly, it is
	$[\tau_{\pi_2}(g(d))]_2 \leq_{\pi_2} [\tau_{\pi_2}(f(d))]_2$ and so $\lambda d. [\tau_{\pi_2}(g(d))]_2 \leq_{\pi_2} \lambda d. [\tau_{\pi_2}(f(d))]_2$.
	We conclude that
	\[
		\tau_\pi(f) = (\lambda d. [\tau_{\pi_2}(f(d))]_1,\lambda d. [\tau_{\pi_2}(f(d))]_2) \preceq_{\pi}
		(\lambda d. [\tau_{\pi_2}(g(d))]_1,\lambda d. [\tau_{\pi_2}(g(d))]_2) = \tau_\pi(g)
	\]

	The proof of the second statement is analogous.

\noindent {\em Statement 3:} Assume that $(f_1,f_2) \preceq_\pi (g_1,g_2)$. Then, for any $d \in \mo{\rho_1}$  it is $f_1(d) \leq_{\pi_2} g_1(d)$
	and $g_2(d) \leq_{\pi_2} f_2(d)$. therefore, $(f_1(d),f_2(d)) \preceq_\pi (g_1(d),g_2(d))$ and by induction it is
	$\tau_\pi^{-1}(f_1(d),f_2(d)) \preceq_{\pi} \tau_\pi^{-1}(g_1(d),g_2(d))$ which by abstracting gives
	$\lambda d. \tau^{-1}_{\pi}(f_1(d), f_2(d)) \preceq_{\pi} \lambda d. \tau^{-1}_{\pi}(g_1(d), g_2(d))$ therefore
	$\tau_\pi^{-1}(f_1,f_2) \preceq_{\pi} \tau_\pi^{-1}(g_1,g_2)$.

	The proof of the fourth statement is analogous.

\noindent {\em Statement 5:} We have:
\[
\begin{array}{rll}
 & \tau^{-1}_{\rho_1 \to \pi_2}(\tau_{\rho_1 \to \pi_2}(f)) \\
 	= & \tau^{-1}_{\rho_1 \to \pi_2}(\lambda d. [\tau_{\pi_2}(f(d))]_1, \lambda d. [\tau_{\pi_2}(f(d))]_2)
	&   \mbox{(Definition of $\tau_{\rho_1 \to \pi_2}$)}\\
	 = & \lambda d. \tau^{-1}_{\pi_2}([\tau_{\pi_2}(f(d))]_1,  [\tau_{\pi_2}(f(d))]_2)
	&   \mbox{(Definition of $\tau^{-1}_{\rho_1 \to \pi_2}$)}\\
	 = & \lambda d. \tau^{-1}_{\pi_2}(\tau_{\pi_2}(f(d)))
	&  \mbox{(Definition of $[\cdot]_1$ and $[\cdot]_2$)}\\
	 = & \lambda d. f(d)
	&   \mbox{(Induction Hypothesis)}\\
	 = & f
\end{array}
\]

\noindent	{\em Statement 6:} Similarly to the previous statement:
\[
\begin{array}{rll}
	&  \tau_{\rho_1 \to \pi_2}(\tau^{-1}_{\rho_1 \to \pi_2}(f_1,f_2)) \\
	= & \tau_{\rho_1 \to \pi_2}(\lambda d. \tau_{\pi_2}^{-1}(f_1(d),f_2(d)))
	&  \mbox{(Definition of $\tau^{-1}_{\rho_1 \to \pi_2}$)}\\
	= & (\lambda d. [\tau_{\pi_2}(\tau^{-1}_{\pi_2}(f_1(d),f_2(d)))]_1, \lambda d. [\tau_{\pi_2}(\tau^{-1}_{\pi_2}(f_1(d),f_2(d)))]_2)
	  & \mbox{(Definition of $\tau_{\rho_1 \to \pi_2}$)}\\
	= & (\lambda d. [(f_1(d),f_2(d))]_1, \lambda d. [(f_1(d),f_2(d))]_2)
	  & \mbox{(Induction Hypothesis)}\\
	= & (\lambda d. f_1(d),\lambda d. f_2(d))
	   & \mbox{(Definition of $[\cdot]_1$ and $[\cdot]_2$)}\\
	= & (f_1,f_2)
\end{array}
\]

The proof extends for the functions $\tau,\tau^{-1}$ defined for interpretations.
For example, for any interpretation ${\cal I} \in {\cal H}_{\mathsf{P}}$ and for any predicate constant $\mathsf{p}$ it follows easily from the previous result and the definitions of $\tau,\tau^{-1}$
that $(\tau^{-1}(\tau({\cal I})))(\mathsf{p}) ={\cal I}(\mathsf{p})$ and
for any $(I,J) \in H^{c}_{\mathsf{P}}$ it is $\tau(\tau^{-1}(I,J))(\mathsf{p}) = (I(\mathsf{p}) ,J(\mathsf{p}))$.
% {\bf TODO: Add $\tau$ of interpretations}
\end{proof}

We proceed now to the second result of Section~\ref{sec:AFT_stable_model_semantics},
Lemma~\ref{ATP_is_approximator_of_TP}, which shows that the mapping $\ATP$ of Definition~\ref{def:ATP} is a \emph{consistent approximator} of $T_\mathsf{P}$, i.e.\ $\ATP$ is $\preceq$-monotonic and for every $I \in H_\mathsf{P}$, $\ATP(I,I) = (\TP(I),\TP(I))$.
The term \emph{consistent} comes from the terminology used by~\citet{DMT04Ultimateapproximationapplicationnonmonotonicknowledgerepresentation},
and it refers to the fact that the mapping is defined over sets of consistent elements, \ie of the form defined in Definition~\ref{def:consistent_lattice}.
A few intermediate lemmas are needed to ease the proof of Lemma~\ref{ATP_is_approximator_of_TP}.
% \vfil
\begin{lemma}\label{tau-1-of-exact}
Let $\pi$ be a predicate type and $f\in \mo{\pi}$. Then $\tau^{-1}_{\pi}(f,f) = f$.
\end{lemma}
\begin{proof}
By induction on $\pi$. If $\pi=o$ it follows from the definition of $\tau^{-1}_{o}$.
When $\pi=\rho \to \pi'$ assuming that lemma holds for $\pi'$, then
$\tau^{-1}_{\rho \to \pi'}(f,f) = \lambda d.\tau^{-1}_{\pi'}(f(d),f(d))=\lambda d.f(d) = f$.
\end{proof}

\begin{corollary}\label{tau-1-of-exact-interpretations}
Let $I \in H_\mathsf{P}$. Then $\tau^{-1}(I,I) = I$.
\end{corollary}

\begin{lemma}\label{two-valued-is-maximal}
Let $\pi$ be a predicate type and $f\in \mo{\pi}$. Then $f$ is $\preceq_\pi$-maximal over $\mos{\pi}$.
\end{lemma}
\begin{proof}
By induction on $\pi$. If $\pi=o$ it follows from the definition of $\preceq_o$.
When $\pi=\rho \to \pi'$ assuming that lemma holds for $\pi'$, then for every $d\in \mo{\rho}$,
$f(d) \in \mo{\pi'}$, so that $f(d)$ is $\preceq_{\pi'}$-maximal over $\mos{\pi'}$.
We conclude that $f$ is $\preceq_{\rho\to\pi'}$-maximal over $\mos{\rho\to\pi'}$.
\end{proof}

\begin{lemma}\label{interlacing}
Let $\mathsf{P}$ be a program and $\pi$ be a predicate type. Let $I$ be a non-empty index-set and for any $i \in I$, $d_i, d'_i \in \mos{\pi}$.
If for all $i\in I$, $d_i \preceq_\pi d'_i$, then $\bigvee_{\leq_\pi}\{d_i \mid i \in I\} \preceq_\pi \bigvee_{\leq_\pi}\{d'_i \mid i \in I\}$.
\end{lemma}
\begin{proof}
We proceed by induction on the predicate type $\pi$.
If $\pi=o$ the lemma follows by case analysis of $\bigvee_{\leq_\pi}\{d_i \mid i \in I\}$.
Suppose now that $\pi = \rho\to\pi'$ and the lemma holds for $\pi'$.
By the proof of Proposition~\ref{semantics_of_types_lattice_cpo}, we have that
$\bigvee_{\leq_\pi}\{d_i \mid i \in I\} = \lambda x. \bigvee_{\leq_{\pi'}}\{d_i(x) \mid i \in I\}$ and
$\bigvee_{\leq_\pi}\{d'_i \mid i \in I\} = \lambda x. \bigvee_{\leq_{\pi'}}\{d'_i(x) \mid i \in I\}$.
For any $x\in \mo{\rho}$, by induction hypothesis, we have that
$\bigvee_{\leq_{\pi'}}\{d_i(x) \mid i \in I\} \preceq_{\pi'} \bigvee_{\leq_{\pi'}}\{d'_i(x) \mid i \in I\}$.
We conclude that $\bigvee_{\leq_\pi}\{d_i \mid i \in I\} \preceq_\pi \bigvee_{\leq_\pi}\{d'_i \mid i \in I\}$.
\end{proof}

\begin{lemma}\label{expressions_fitting}
Let $\mathsf{P}$ be a program, let ${\cal I},{\cal J}\in {\cal H}_\mathsf{P}$, and let $s$ be a Herbrand state of $\mathsf{P}$.
For every expression $\mathsf{E} : \pi$, if ${\cal I} \preceq {\cal J}$ then $\mwrst{\mathsf{E}}{{\cal I}}{s} \preceq_\pi \mwrst{\mathsf{E}}{{\cal J}}{s}$.
\end{lemma}
\begin{proof}
Using induction on $\mathsf{E}$. The only interesting case is when $\mathsf{E} = (\mathsf{E_1}~\mathsf{E_2})$ where $\mathsf{E}_1 : \rho \to \pi$ and $\mathsf{E}_2 : \rho$.
Suppose that when $\mathcal{I}\preceq \mathcal{J}$ then $\mwrst{\mathsf{E_1}}{{\cal I}}{s} \preceq_{\rho\to\pi} \mwrst{\mathsf{E_1}}{{\cal J}}{s}$ and $\mwrst{\mathsf{E_2}}{{\cal I}}{s} \preceq_{\rho} \mwrst{\mathsf{E_2}}{{\cal J}}{s}$.
Suppose $x\in \{\lsem \mathsf{E}_1 \rsem^{*}_s(\mathcal{J})(d) \mid d \in \lsem \rho\rsem, \lsem \mathsf{E}_2 \rsem^{*}_s({\cal J}) \preceq_{\rho} d\}$.
Then, there exists some $d$ such that $\lsem \mathsf{E}_2 \rsem^{*}_s({\cal J}) \preceq_{\rho} d$ and $x = \lsem \mathsf{E}_1 \rsem^{*}_s(\mathcal{J})(d)$.
Since $\mwrst{\mathsf{E_2}}{{\cal I}}{s} \preceq_{\rho} \mwrst{\mathsf{E_2}}{{\cal J}}{s} \preceq_{\rho} d$, we have
$\lsem \mathsf{E}_1 \rsem^{*}_s(\mathcal{I})(d) \in \{\lsem \mathsf{E}_1 \rsem^{*}_s(\mathcal{I})(d) \mid d \in \lsem \rho\rsem, \lsem \mathsf{E}_2 \rsem^{*}_s({\cal I}) \preceq_{\rho} d\}$.
Also, by inductive hypothesis, $\lsem \mathsf{E}_1 \rsem^{*}_s(\mathcal{I})(d) \preceq_{\pi} x$, so that
$\bigwedge_{\preceq_\pi}\{\lsem \mathsf{E}_1 \rsem^{*}_s(\mathcal{I})(d) \mid d \in \lsem \rho\rsem, \lsem \mathsf{E}_2 \rsem^{*}_s({\cal I}) \preceq_{\rho} d\} \preceq_{\pi} x$.
Since that holds for any $x$, we have $\mwrst{\mathsf{E}}{{\cal I}}{s} =
\bigwedge_{\preceq_\pi}\{\lsem \mathsf{E}_1 \rsem^{*}_s(\mathcal{I})(d) \mid d \in \lsem \rho\rsem, \lsem \mathsf{E}_2 \rsem^{*}_s({\cal I}) \preceq_{\rho} d\} \preceq_\pi
\bigwedge_{\preceq_\pi}\{\lsem \mathsf{E}_1 \rsem^{*}_s(\mathcal{J})(d) \mid d \in \lsem \rho\rsem, \lsem \mathsf{E}_2 \rsem^{*}_s({\cal J}) \preceq_{\rho} d\} =
\mwrst{\mathsf{E}}{{\cal J}}{s}$.
\end{proof}

We are finally ready to prove Lemma \ref{ATP_is_approximator_of_TP}.
% \vfil
\begin{relemma}{ATP_is_approximator_of_TP} %\samuele{Uses:  Lemma~\ref{expressions_fitting} and Lemma~\ref{interlacing}, Proposition~\ref{tau-isomorphism-preserves}, Lemma~\ref{two-valued-is-maximal}, Lemma~\ref{tau-1-of-exact}}
Let $\mathsf{P}$ be a program. In the terminology of \citet{DMT04Ultimateapproximationapplicationnonmonotonicknowledgerepresentation}, $\ATP: H^{c}_\mathsf{P} \to H^{c}_\mathsf{P}$
is a consistent
%\bart{I think just saying ``approximator'' is fine}
approximator of $T_\mathsf{P}$.
\end{relemma}
\begin{proof}
We have to show that $\ATP$ is $\preceq$-monotone and extends $\TP$.
For the monotonicity, it follows from the definition of ${\cal T}_{\mathsf{P}}$ together with Lemma~\ref{expressions_fitting} and Lemma~\ref{interlacing}
that ${\cal T}_{\mathsf{P}}$ is $\preceq$-monotone. Also, by Proposition~\ref{tau-isomorphism-preserves}, $\tau$ and $\tau^{-1}$ preserve $\preceq$,
so that $\ATP$ is $\preceq$-monotone.

Now, we have to show that $\ATP$ extends $\TP$, \ie for every $I \in H_\mathsf{P}$, $\ATP(I,I) = (\TP(I),\TP(I))$.
By Corollary~\ref{tau-1-of-exact-interpretations}, $\tau^{-1}(I,I) = I$.
Since $I \in H_\mathsf{P}$, by Lemma~\ref{mo-mos-coincide}, we have that for every expression $E$, $\mo{E}_s(\tau^{-1}(I,I)) = \mos{E}_s(I)$.
Now we have
\[
\begin{array}{rll}
  \ATP(I,I)
 	= & \tau({\cal T}_{\mathsf{P}}(\tau^{-1}(I,I))) \\
	= & \tau\big(\bigvee_{\leq_\bool}\{\mwrst{\mathsf{B}}{\tau^{-1}(I,I)}{s[\overline{\mathsf{R}}/\overline{d}]} \mid \mbox{$s\in S_{\mathsf{P}}$ and $(\mathsf{p}\ \overline{\mathsf{R}} \lrule \mathsf{B})$ in $\mathsf{P}$}\}\big)\\
	= & \tau\big(\bigvee_{\leq_\bool}\{\lsem \mathsf{B} \rsem_{s[\overline{\mathsf{R}}/\overline{d}]}(I) \mid \mbox{$s\in S_{\mathsf{P}}$ and $(\mathsf{p}\ \overline{\mathsf{R}} \lrule \mathsf{B})$ in $\mathsf{P}$}\}\big)\\
  = & \tau(\TP(I))
\end{array}
\]
Since $\TP(I) \in H_\mathsf{P}$, by Corollary~\ref{tau-1-of-exact-interpretations}, we have $\tau(\TP(I)) = (\TP(I), \TP(I))$.
\end{proof}

% \begin{proof}
% %
% First notice that, by Lemma~\ref{rel-semantics} and the definitions of $T_{\mathsf{P}}$ and $\AP$, it holds that $\AP(I,I)=(T_{\mathsf{P}}(I),T_{\mathsf{P}}(I))$. It remains to show that $\AP$ is $\leq_k$-monotone. Consider $(I_1,I_2),(I'_1,I'_2)$ such that $(I_1,I_2) \leq_k (I'_1,I'_2)$, ie., $I_1 \leq I'_1$ and $I'_2 \leq I_2$. We show that $\AP(I_1,I_2)\leq_k \AP(I'_1,I'_2)$. 	Let
% $\mathsf{p} : \rho_1 \rightarrow \cdots \rightarrow \rho_n \rightarrow o$ be a predicate constant, and
% $\overline{d} \in \mo{\rho_1} \times \cdots \times \mo{\rho_n}$. By Lemma \ref{k-mon-literals}, for all $s\in S_{\mathsf{P}}$ and
% $(\mathsf{p}\ \overline{\mathsf{V}} \leftarrow \mathsf{B})$ in $\mathsf{P}$, it holds that $\mwrst{\mathsf{B}}{I_1,I_2}{s[\overline{\mathsf{V}}/\overline{d}]}\leq_k \mwrst{\mathsf{B}}{I_1',I_2'}{s[\overline{\mathsf{V}}/\overline{d}]}$. By Definition \ref{def:approximator}, it follows that $\AP(I_1,I_2)(\mathsf{p})\ \overline{d}\leq_k \AP(I'_1,I'_2)(\mathsf{p})\ \overline{d}$. By the arbitrarity of $\mathsf{p}$ and $ \overline{d}$, we have $\AP(I_1,I_2)\leq_k \AP(I'_1,I'_2)$, as desired.
% %
% \end{proof}

We conclude this appendix by showing that Definition~\ref{def:AFTsemantics} and Definition~\ref{def:three-valued_model} for three-valued models agree.

\begin{relemma}{model-iff-atp-prefixpoint}
Let $\mathsf{P}$ be a program and $(I, J)\in H^{c}_\mathsf{P}$.
Then, $(I, J)$ is a pre-fixpoint of $\ATP$ if and only if $\tau^{-1}(I, J)$ is a three-valued model of $\mathsf{P}$.
\end{relemma}
\begin{proof}
	%We denote by $\mathcal{I}$ the three-vaued interpretation $\tau^{-1}(I,J)$.
First notice that by Definition \ref{def:three-valuedTP} and Proposition
\ref{tau-isomorphism-preserves}, $(I, J)$ is a pre-fixpoint of $\ATP$ if and
only if $\ATP(I,J) = \tau({\cal T}_\mathsf{P}(\tau^{-1}(I,J)))\leq (I,J)$ if
and only if ${\cal T}_\mathsf{P}(\tau^{-1}(I,J))\leq \tau^{-1}(I,J)$, i.e.\
$(I,J)$ is a pre-fixpoint of $\ATP$ if and only if $\tau^{-1}(I,J)$ is a
fixpoint of ${\cal T}_\mathsf{P}$. We conclude by Proposition
\ref{three-valued-model-iff-tp-prefixpoint}.
\end{proof}

% \begin{relemma}{exact-fixpoints-are-models}
% %
% Let $(I,I)$ be a fixpoint of $\ATP$. Then, $I$ is a two-valued
% model of $\mathsf{P}$.
% %
% \end{relemma}
% \begin{proof}
% By Lemma~\ref{ATP_is_approximator_of_TP}, $\ATP$ extends $\TP$.
% Therefore, $(I,I) = \ATP(I,I) = (\TP(I), \TP(I))$. So, $I$ is a fixpoint of $\TP$.
% Using Proposition~\ref{model-iff-tp-prefixpoint}, $I$ is a two-valued model of $\mathsf{P}$.
% \end{proof}
% \pagebreak

\section{Proofs of Section~\ref{sec:properties}}

The following theorem states that our stable model semantics coincides with the classical stable model semantics
for the class of propositional programs.
\begin{retheorem}{coincides-with-classical-stable-models}
Let $\mathsf{P}$ be a propositional logic program. Then, ${\cal M}$ is a (three-valued)
stable model of $\mathsf{P}$ iff ${\cal M}$ is a classical (three-valued) stable model of $\mathsf{P}$.
\end{retheorem}

\begin{proof}
In \cite[Section 6, pages 107--108]{DMT04Ultimateapproximationapplicationnonmonotonicknowledgerepresentation}, the well-founded semantics of propositional logic programs
is derived. The language used there allows arbitrary nesting of conjunction, disjunction and negation in bodies of the rules
which fully encompasses our syntax when we restrict our programs to be propositional. In addition, the immediate consequence operator $\TP$ is the same and so is
the approximation space which we have denoted as $H^{c}_{\mathsf{P}}$ in this work.

It is easy to see that the approximator $\ATP$ we give in our approach and the one given in \cite{DMT04Ultimateapproximationapplicationnonmonotonicknowledgerepresentation}  fully coincide for propositional programs therefore produce equivalent
semantics. Notice how our three-valued operator ${\cal T}_{\mathsf{P}}$ fully coincides with the three-valued immediate consequence operator in
\cite{DMT04Ultimateapproximationapplicationnonmonotonicknowledgerepresentation} since the fourth rule in Definition \ref{tuple-semantics} is never used.
\end{proof}

In order to establish Theorem~\ref{stable-models-are-minimal}, we use the following
proposition which is a restatement of Proposition~3.14 found in \cite{DMT04Ultimateapproximationapplicationnonmonotonicknowledgerepresentation}
that refers to pre-fixpoints of the approximator $\ATP$ instead of fixpoints. The proof is almost identical but is presented here, nonetheless, for reasons of completeness.
\begin{proposition}\label{stable-fixpoints-minimal-prefixpoints}
	A stable fixpoint $(x,y)$ of $\ATP$ is a $\leq$-minimal pre-fixpoint of $\ATP$. Furthermore, if $(x,x)$
	is a stable fixpoint of $\ATP$ then $x$ is a minimal pre-fixpoint of $\TP$.
\end{proposition}

\begin{proof}
Let $(x,y)$ be a stable fixpoint of $\ATP$ and let $(x',y')$ such that $(x',y') \leq(x,y)$ and $(x',y')$ is a pre-fixpoint of $\ATP$, so
$\ATP(x',y') \leq (x',y')$. We have that $x' \leq y' \leq y$ which gives us that
$\ATP(x',y)_1 \leq \ATP(x',y')_1 \leq (x',y')_1 = x'$. Therefore, $x'$ is a pre-fixpoint of the operator $\ATP(\cdot,y)_1$ and since
$x$ is its least fixpoint we get that $x \leq x'$. By the assumption that $x' \leq x$ we conclude that $x=x'$.
% \newline

Since we have shown that $x =x'$ we have that $x =x' \leq y'$ and $\ATP(x,y')_2 \leq \ATP(x',y')_2 \leq (x',y')_2 = y'$ which makes $y'$ a pre-fixpoint of
$\ATP(x, \cdot)_2$. Since $y$ is its least fixpoint we have that $y \leq y'$ and by assumption $y' \leq y$. We conclude that $y=y'$  and
finally $(x,y)=(x',y')$.

Assume that $x' \leq x$ and $x'$ is a pre-fixpoint of $\TP$ therefore $\TP(x')\leq x'$. Since $\ATP$ is an approximator of $\TP$ we have that
$\ATP(x',x') =(\TP(x'),\TP(x')) \leq (x',x')$. But then $(x',x')$ is a pre-fixpoint of $\ATP$ and $(x',x')\leq (x,x)$. By the result of the previous
paragraph we conclude that $(x',x')=(x,x)$ and $x'=x$.
\end{proof}
% \pagebreak

\begin{retheorem}{stable-models-are-minimal}
%
%Let $\mathsf{P}$ be a $\HOL$ program.
All (three-valued) stable models of a $\HOL$ program $\mathsf{P}$ are
%then ${\cal M}$ is a
$\leq$-minimal models of~$\mathsf{P}$.
% If ${\cal M}$ is a higher-order three-valued
% stable model of $\mathsf{P}$, then ${\cal M}$ is a $\leq$-minimal three-valued model of $\mathsf{P}$.
%
\end{retheorem}
\begin{proof}
Let $\cal M$ be a three-valued stable model of $\mathsf{P}$ and $\cal M'$ a three-valued model of $\mathsf{P}$, such that ${\cal M'} \leq {\cal M}$ which also implies
$\tau(\cal M') \leq \tau(\cal M)$. By Lemma~\ref{model-iff-atp-prefixpoint}, $\tau(\cal M')$ is a pre-fixpoint of $\ATP$.
But $\tau(\cal M)$ is a stable fixpoint of $\ATP$ so by Proposition~\ref{stable-fixpoints-minimal-prefixpoints}
it is a minimal pre-fixpoint of $\ATP$. We conclude that $\tau(\cal M')=\tau(\cal M)$ and so $\cal M'=\cal M$.
Since every two-valued model is a three-valued model it follows that
that every stable model is also a $\leq$-minimal two-valued model.
\end{proof}

%
% Assume a $\cal M'$ such that $\cal M' \leq \cal M$ and $\cal M'$ a three-valued model of $\mathsf{P}$. This also implies that $\tau(\cal M') \leq \tau(\cal M$).
% By Lemma \ref{model-iff-atp-prefixpoint}, $\tau(\cal M')$ is a pre-fixpoint of $\ATP$.
% But $\tau(\cal M)$ is a stable fixpoint of $\ATP$ so by Proposition \ref{stable-fixpoints-minimal-prefixpoints}
% it is a minimal pre-fixpoint of $\ATP$. We conclude that $\tau(\cal M')=\tau(\cal M)$ and so $\cal M'=\cal M$.

% Since $\cal M$ is a stable model, $\cal M$ is two-valued and by corollary \ref{tau-1-of-exact-interpretations}
% it is $\tau({\cal M})=({\cal M},{\cal M})$ and $({\cal M},{\cal M})$ a stable fixpoint of $\ATP$.
% Assume $x$ a two-valued model such that $x \leq {\cal M}$. Then by Proposition \ref{model-iff-tp-prefixpoint} $x$ is a pre-fixpoint of $\TP$.
% By Proposition \ref{stable-fixpoints-minimal-prefixpoints} we conclude that $x={\cal M}$.
% of $\TP$ and $\ATP(x,x)=(\TP(x),\TP(x))\leq (x,x)$ which makes $(x,x)$  a pre-fixpoint of $\ATP$.
% Similarly, assume $M'$ such that $M' \leq M$ and $M'$ is a model of $\mathsf{P}$. Then by Proposition \ref{model-iff-tp-prefixpoint} $M'$ is a pre-fixpoint
% of $\TP$ and $\ATP(M',M')=(\TP(M'),\TP(M'))\leq (M',M')$ which makes $(M',M')$  a pre-fixpoint of $\ATP$. By Proposition \ref{stable-fixpoints-minimal-prefixpoints}
% we conclude that $(M',M')=(M,M)$ and therefore $M'=M$.
%%
\begin{retheorem}{exact-wf-unique-stable}
Let $\mathsf{P}$ be a ${\cal HOL}$ program. If the well-founded model of
$\mathsf{P}$ is two-valued, then this is also its unique stable model.
\end{retheorem}
\begin{proof}
Let $\cal M$ be the well founded model of $\mathsf{P}$. Then $\tau({\cal M})$ is the $\preceq$-least three-valued stable model.
It immediately follows that since $\cal M$ is two-valued, by Corollary \ref{tau-1-of-exact-interpretations} it is $\tau({\cal M})=(\cal M,\cal M)$.
Then for any $(x,y)$ three-valued stable fixpoint of $\ATP$ it is $({\cal M},{\cal M}) \preceq (x,y)$.
Since it also must hold $x \leq y $ we conclude that $x=y=\cal M$ and $\tau^{-1}(x,y)=\tau^{-1}({\cal M},{\cal M})=\cal M$.
\end{proof}

In order to establish Theorem~\ref{well-founded-exact-on-stratified}, we first prove some auxiliary results.
\begin{lemma}\label{stratified-sem}
Let $\mathsf{P}$ be a stratified $\HOL$ program and $\mathsf{E}$ be an expression.
Let $I, J \in H_\mathsf{P}$ be two interpretations such that $I(\mathsf{p})=J(\mathsf{p})$ for every predicate constant $\mathsf{p}$ occurring in $\mathsf{E}$.
Then, for every state $s \in S_\mathsf{P}$, $\lsem \mathsf{E} \rsem_s(I) = \lsem \mathsf{E} \rsem_s(J)$.
\end{lemma}
\begin{proof}
Trivial using induction on the structure of $\mathsf{E}$.
\end{proof}

\begin{corollary}\label{stratified-Tp}
Let $S$ be a stratification function of the $\HOL$ program $\mathsf{P}$
% Let $\mathsf{P}$ be a stratified program and suppose $S$ is its stratification function.
and $I, J \in H_\mathsf{P}$ be two interpretations.
If for some $n\in\omega$, $I(\mathsf{p})=J(\mathsf{p})$ for every predicate constant $\mathsf{p}$ with $S(\mathsf{p})< n$,
then $\TP(I)(\mathsf{p}) = \TP(J)(\mathsf{p})$ for every predicate constant $\mathsf{p}$ with $S(\mathsf{p})< n$.
\end{corollary}

\begin{lemma}\label{stratified-A-eq-O-O}
Let $S$ be a stratification function of the $\HOL$ program $\mathsf{P}$
and $(I, J) \in H^{c}_\mathsf{P}$.
If for some $n\in\omega$, $I(\mathsf{p})=J(\mathsf{p})$ for every predicate constant $\mathsf{p}$ with $S(\mathsf{p})<n$,
then $\ATP(I, J)(\mathsf{p}) = (\TP(I)(\mathsf{p}), \TP(J)(\mathsf{p}))$ for every predicate constant $\mathsf{p}$ with $S(\mathsf{p})\leq n$.
\end{lemma}
\begin{proof}
We will show the following auxiliary statement that suffices to show the lemma.
For any expression $\mathsf{E} : \pi$ such that the following three statements hold:
\begin{enumerate}
  \item $S(\mathsf{q})\leq n$ for every predicate constant $\mathsf{q}$ occurring in $\mathsf{E}$,
  \item if $\mathsf{E}$ is of the form $(\mathsf{E}_1~\mathsf{E}_2)$, then $S(\mathsf{q})< n$ for every predicate constant $\mathsf{q}$ occurring in $\mathsf{E}_2$,
  \item if $\mathsf{E}$ is of the form $(\sim \mathsf{E}_1)$, then $S(\mathsf{q})< n$ for every predicate constant $\mathsf{q}$ occurring in $\mathsf{E}_1$,
\end{enumerate}
and for any Herbrand state $s \in S_\mathsf{P}$ it follows that $\lsem \mathsf{E} \rsem^*_s(\tau^{-1}(I,J)) = \tau^{-1}_\pi(\lsem \mathsf{E} \rsem_s(I),\lsem \mathsf{E} \rsem_s(J))$.
This can be established using induction on the structure of $\mathsf{E}$.
The interesting cases are when $\mathsf{E}$ is of the form $(\mathsf{E_1}~\mathsf{E_2})$ or of the form $(\sim\mathsf{E_3})$.
Suppose that $\mathsf{E} : \pi$ is of the form $(\mathsf{E_1}~\mathsf{E_2})$ where $\mathsf{E}_1 : \rho\to\pi$ and $\mathsf{E}_2 : \rho$
and suppose the statement holds for $\mathsf{E_1}$ and $\mathsf{E_2}$.
By Lemma~\ref{stratified-sem}, $\lsem \mathsf{E}_2 \rsem_s(I) = \lsem \mathsf{E}_2 \rsem_s(J)$.
So, by the induction hypothesis and Lemma~\ref{tau-1-of-exact}, $\lsem \mathsf{E}_2 \rsem^*_s(\tau^{-1}(I,J)) = \lsem \mathsf{E}_2 \rsem_s(I)$ and
therefore $\lsem \mathsf{E}_2 \rsem^*_s(\tau^{-1}(I,J))\in \mo{\rho}$.
By Lemma~\ref{two-valued-is-maximal},
$\bigwedge_{\preceq_{\pi}}\{\lsem \mathsf{E}_1 \rsem^{*}_s(\tau^{-1}(I,J))(d) \mid d \in \lsem \rho\rsem, \lsem \mathsf{E}_2 \rsem^{*}_s(\tau^{-1}(I,J)) \preceq_{\rho} d\} = \lsem \mathsf{E}_1 \rsem^{*}_s(\tau^{-1}(I,J))(\lsem \mathsf{E}_2 \rsem^*_s(\tau^{-1}(I,J))) = \lsem \mathsf{E}_1 \rsem^{*}_s(\tau^{-1}(I,J))(\lsem \mathsf{E}_2 \rsem_s(I))$.
Therefore,
\[
\begin{array}{rll}
    & \lsem \mathsf{E} \rsem^*_s(\tau^{-1}(I,J))\\
 	 = & \lsem \mathsf{E}_1 \rsem^{*}_s(\tau^{-1}(I,J))(\lsem \mathsf{E}_2 \rsem_s(I))\\
	 = & \tau^{-1}_{\rho\to\pi}(\lsem \mathsf{E}_1 \rsem_s(I),\lsem \mathsf{E}_1 \rsem_s(J))(\lsem \mathsf{E}_2 \rsem_s(I))
	    & \mbox{(Induction Hypothesis)}\\
   = & (\lambda d.\tau^{-1}_{\pi}(\lsem \mathsf{E}_1 \rsem_s(I)(d),\lsem \mathsf{E}_1 \rsem_s(J)(d)))(\lsem \mathsf{E}_2 \rsem_s(I))
      & \mbox{(Definition of $\tau^{-1}_{\rho\to\pi}$)}\\
   = & \tau^{-1}_{\pi}(\lsem \mathsf{E}_1 \rsem_s(I)(\lsem \mathsf{E}_2 \rsem_s(I)),\lsem \mathsf{E}_1 \rsem_s(J)(\lsem \mathsf{E}_2 \rsem_s(I)))\\
   = & \tau^{-1}_{\pi}(\lsem \mathsf{E}_1 \rsem_s(I)(\lsem \mathsf{E}_2 \rsem_s(I)),\lsem \mathsf{E}_1 \rsem_s(J)(\lsem \mathsf{E}_2 \rsem_s(J)))
      &  \mbox{(Since $\lsem \mathsf{E}_2 \rsem_s(I) = \lsem \mathsf{E}_2 \rsem_s(J)$)}\\
   = & \tau^{-1}_{\pi}(\lsem \mathsf{E} \rsem_s(I),\lsem \mathsf{E} \rsem_s(J))
\end{array}
\]
Now, suppose that $\mathsf{E} : \bool$ is of the form $(\sim \mathsf{E_3})$ where $\mathsf{E}_3: \bool$.
By Lemma~\ref{stratified-sem}, $\lsem \mathsf{E}_3 \rsem_s(I) = \lsem \mathsf{E}_3 \rsem_s(J)$.
So, by Lemma~\ref{tau-1-of-exact}, $\lsem \mathsf{E}_3 \rsem^*_s(\tau^{-1}(I,J)) = \lsem \mathsf{E}_3 \rsem_s(I)$.
Since $\mathsf{E}_3$ is of type $\bool$, $\lsem \mathsf{E}_3 \rsem^*_s(\tau^{-1}(I,J))$ can be either $\mathit{true}$ or $\mathit{false}$.
In any of the two cases, it is easy to show that $\lsem \mathsf{E} \rsem^*_s(\tau^{-1}(I,J)) = \tau^{-1}_\bool(\lsem \mathsf{E} \rsem_s(I),\lsem \mathsf{E} \rsem_s(J))$.
\end{proof}

\begin{retheorem}{well-founded-exact-on-stratified}
Let $\mathsf{P}$ be a stratified $\HOL$ program. Then, the well-founded model of $\mathsf{P}$ is two-valued.
\end{retheorem}
\begin{proof}
Let $S$ be a stratification function of $\mathsf{P}$ and $(I_w, J_w) \in H^{c}_\mathsf{P}$ be the well-founded model of $\mathsf{P}$.
Suppose, for the sake of contradiction, that $I_w \neq J_w$. Let $n$ be the least number such that there exists some
predicate constant $\mathsf{p} : \pi_1$ with $S(\mathsf{p})=n$ and $I_w(\mathsf{p})\neq J_w(\mathsf{p})$.
We define an interpretation $J$ such that for any predicate constant $\mathsf{q}$:
\[
  J(\mathsf{q}) = \begin{cases}
    I_w(\mathsf{q}), &\text{if } S(\mathsf{q})\leq n \\
    J_w(\mathsf{q}), &\text{if } S(\mathsf{q})>n
  \end{cases}
\]
It is obvious, by definition, that $I_w \leq J \leq J_w$ and therefore $(I_w, J) \in H^{c}_\mathsf{P}$.
For any predicate constant $\mathsf{q}$ with $S(\mathsf{q})\leq n$, we have
\[
  \begin{array}{rlll}
    [\ATP(I_w, J)]_2(\mathsf{q}) & = & \TP(J)(\mathsf{q}) & \mbox{(Lemma~\ref{stratified-A-eq-O-O})} \\
    & = & \TP(I_w)(\mathsf{q}) & \mbox{(Corollary~\ref{stratified-Tp})} \\
    & = & [\ATP(I_w, J_w)]_1(\mathsf{q}) & \mbox{(Lemma~\ref{stratified-A-eq-O-O})} \\
    & = & I_w(\mathsf{q}) & \mbox{($(I_w, J_w)$ is a fixpoint of $\ATP$)} \\
    & = & J(\mathsf{q}) & \mbox{(Definition of $J$)}
  \end{array}
\]

Since $[\ATP(I_w, \cdot)]_2$ is monotone and $J\leq J_w$, it follows
$[\ATP(I_w, J)]_2 \leq [\ATP(I_w, J_w)]_2 = J_w$.
Thus, for any predicate constant $\mathsf{q} : \pi_2$ with $S(\mathsf{q})>n$,
we have $[\ATP(I_w, J)]_2(\mathsf{q}) \leq_{\pi_2} J_w(\mathsf{q}) = J(\mathsf{q})$.
Since $[\ATP(I_w, J)]_2(\mathsf{q}) \leq_{\pi_2} J(\mathsf{q})$ for any predicate constant $\mathsf{q} : \pi_2$,
we have $[\ATP(I_w, J)]_2 \leq J$, or $J$ is a pre-fixpoint of $[\ATP(I_w, \cdot)]_2$.
Since $(I_w, J_w)$ is a stable fixpoint of $\ATP$,
$J_w$ is the least pre-fixpoint of $[\ATP(I_w, \cdot)]_2$. Therefore, we have $J_w \leq J$.
So, $J_w(\mathsf{p}) \leq_{\pi_1} J(\mathsf{p}) = I_w(\mathsf{p})$. But, we have $I_w \leq J_w$,
so that $I_w(\mathsf{p}) = J_w(\mathsf{p})$, which is a contradiction.
We conclude that $I_w = J_w$. Using Corollary~\ref{tau-1-of-exact-interpretations},
$\tau^{-1}(I_w, J_w)= I_w$, so that $\tau^{-1}(I_w, J_w) \in H_\mathsf{P}$.
\end{proof}

\fi
\end{document}